\documentclass[journal,12pt,onecolumn,compsoc]{IEEEtran}

\usepackage[T1]{fontenc}
\usepackage[utf8]{inputenc}
\usepackage{lmodern}

\usepackage[pdftex]{graphicx}
\usepackage{color}
\usepackage{longtable}

\usepackage[cmex10]{amsmath}
\usepackage{amssymb,amsthm}

\usepackage{algorithmic}
\usepackage{array}
\usepackage[caption=false,font=footnotesize]{subfig}
\usepackage{url}
\usepackage{enumitem}
\usepackage{multicol}
\usepackage{caption}
\usepackage[bb=boondox]{mathalfa}

\newcommand{\node}[1]{\mathrm{#1}}

\newcommand{\ds}{\displaystyle}
\newcommand{\ts}{\textstyle}

\newcommand{\Nn}{{\mathbb N}}

\newcommand{\Rr}{{\mathbb R}}

\newcommand{\normconst}{\|(\mathbf{S}_W \mathbf{B}_M\!)^{\minus 1}\|  }
\newcommand{\itpx}{\mathrm{I}_{W} x}
\newcommand{\funit}{f_{\mathbb{1}}}
\newcommand{\zero}{\mathbb{0}}

\theoremstyle{plain}
\newtheorem{theorem}{Theorem}[section]
\newtheorem{proposition}[theorem]{Proposition}
\newtheorem{corollary}[theorem]{Corollary}
\newtheorem{lemma}[theorem]{Lemma}

\theoremstyle{remark}
\newtheorem{remark}[theorem]{Remark}

\theoremstyle{definition}
\newtheorem{definition}[theorem]{Definition}

\usepackage{tikz}
\usetikzlibrary{calc,fadings,decorations.pathreplacing}
\usetikzlibrary{matrix,arrows}

\def\minus{%
  \setbox0=\hbox{-}%
  \vcenter{%
    \hrule width\wd0 height \the\fontdimen8\textfont3%
  }%
}

\usepackage[ruled,vlined]{algorithm2e}

\begin{document}

\title{Graph signal interpolation with \\ Positive Definite Graph Basis Functions}

\author{Wolfgang Erb
\thanks{Universit{\`a} degli Studi di Padova, Dipartimento di Matematica ''Tullio Levi-Civita'', wolfgang.erb@lissajous.it.}
}

\markboth{Graph signal interpolation with Positive Definite Functions}%
{Graph signal interpolation with Positive Definite Functions}

\maketitle

\begin{abstract}
For the interpolation of graph signals with generalized shifts of a graph basis function (GBF), we introduce the concept of positive definite functions on graphs. This concept merges kernel-based interpolation with spectral theory on graphs and can be regarded as a graph analog of radial basis function interpolation in euclidean spaces or spherical basis functions. We provide several descriptions of positive definite functions on graphs, the most relevant one is a Bochner-type characterization in terms of positive Fourier coefficients. These descriptions allow us to design GBF's and to study GBF interpolation in more detail: we are able to characterize the native spaces of the interpolants, we provide explicit estimates for the interpolation error and obtain bounds for the numerical stability. As a final application, we show how GBF interpolation can be used to get quadrature formulas on graphs. 
\end{abstract}


\IEEEpeerreviewmaketitle

\section{Introduction}

\IEEEPARstart{G}{raph} signal processing is a rapidly growing research field for the study of big data structures on highly irregular and complex graph domains \cite{Ortega2018,Sandryhaila2013,StankovicDakovicSejdic2019}. In our brave new world such data structures are generated, collected and magnified in every facet of our lives: in social networks, in our health systems, in banking and shopping apps, in traffic or security monitoring. Graphs offer the possibility to model, connect and order these structures, and graph signal processing offers the tools to filter and simplify the data on the graphs as well as to extract the most relevant information.

\begin{figure}[htbp]
	\centering
	\includegraphics[width= 1\textwidth]{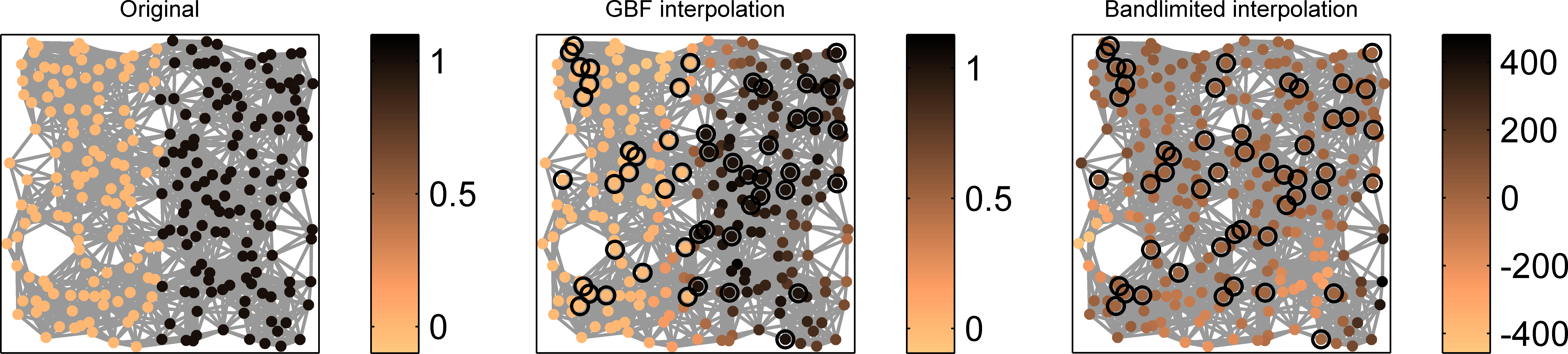} 
	\caption{Comparison of GBF interpolation to bandlimited spectral interpolation on a sensor graph. Left: original signal. Middle: GBF interpolated signal. Right: Bandlimited interpolation of the signal with strong Runge-tpye artifacts. The data samples are taken on the ringed nodes of the graph. }
	\label{fig:comparisoninterpolation}
\end{figure}

For signals on highly complex graphs, interpolation and approximation methods are essential tools  to reduce the computational costs or to reconstruct signals from a small amount of measurements. Inspired from classical signal processing on the real line, spaces of bandlimited signals have been introduced on graphs to approximate signals \cite{BelkinNiyogi2004}. As in classical Fourier analysis, the main conception is that smooth signals are approximated well by their low-frequency parts, while the removed higher frequencies mostly contain noise. Spaces of bandlimited signals can therefore be considered as natural approximation spaces on graphs. For the interpolation of graph signals, the usage of bandlimited signals has however some drawbacks:
\begin{itemize} 
\item[(I1)] For a given set of interpolation nodes on the graph, \emph{uniqueness of interpolation} can not be guaranteed in the space of bandlimited functions, cf.  \cite{Pesenson2008}.
\item[(I2)] More critical for applications, also if unisolvence is given, the bandlimited interpolants show an unstable and highly oscillatory behavior, particularly in the boundary regions of the graph. In analogy to a similar behavior in high-order polynomial interpolation, this can be considered as a \emph{Runge-type phenomenon}. An example of such a Runge artifact is shown in Figure \ref{fig:comparisoninterpolation} (right). 
\end{itemize}   
Similar as for high-order polynomial interpolation, there are however some possibilities to circumnavigate the issues (I1) and (I2): 
\begin{itemize}
\item[(S1)] \emph{Adaptive selection of optimal sampling nodes.} This strategy has a stabilizing effect on the interpolation if the sampling nodes can be chosen freely. For the adaptive refinement typically a costly Greedy algorithm is used. Realizations of this strategy on graphs can be found in \cite{Chenetal2015,Narang2013,TBL2016}.
\item[(S2)] \emph{Regularization of the interpolation conditions.} If the interpolation condition is weakened to the solution of a regression problem the target space of bandlimited signals can be reduced and the unisolvence (I1) guaranteed \cite{BelkinNiyogi2004,Narang2013}. Alternatively, if the problem is formulated as a minimization problem with additional smoothness or sparsity constraints, cf. \cite{StankovicDakovicSejdic2019,TBL2016}, unisolvence is obtained and the Runge artifacts described in (I2) are in general mitigated. 
\end{itemize} 

Beyond the strategies (S1) and (S2), a much more flexible tool for the reconstruction of signals from a small set of samples is given by kernel-based methods. On graphs, such kernel methods are usually studied in terms of regression and regularization techniques for machine learning, see \cite{Belkin2006,BelkinNiyogi2004,KondorLafferty2002,Romero2017,SmolaKondor2003}. Typically used kernels are powers of the graph Laplacian \cite{BelkinNiyogi2004} and diffusion kernels \cite{KondorLafferty2002,SmolaKondor2003}.  
Works related to graph signal interpolation focused on variational splines as kernels, see \cite{Pesenson2009,Ward2018interpolating}.

For interpolation on graphs, a key advantage of kernel-based methods is the fact that unisolvence (I1) is automatically given as soon as the applied kernel is positive definite. As particular kernels can be constructed to mimic interpolation in bandlimited spaces, (I2) plays a role for kernel-based methods as well. However, the large flexibility to generate different kernels allows to design interpolation kernels that are adapted to  the given data and that avoid reconstruction artifacts. An example of a diffusion kernel based graph interpolation is demonstrated in Figure \ref{fig:comparisoninterpolation} (middle).

The determination of suitable interpolation kernels is essential for the quality in signal reconstruction. As the space composed of all linear kernels is growing quadratically in the number of graph nodes, it is important to focus on simpler classes of kernels that provide suitable smoothness properties for the interpolation and that can be aligned with the spectral structure of the graph. In euclidean spaces such a class of kernels is given by positive definite radial basis functions (RBF's) \cite{buhmann2003,Scha98,We05}. A radial basis function $f$ generates naturally a symmetric kernel $K$ on $\Rr^d$ by taking the shifts $K(x,y) = f(x-y)$ of the function $f$. This allows to represent the kernel $K$ very compactly in terms of a univariate function. Similar concepts of positive definite basis functions exist in other group settings, as for instance for periodic functions \cite{Schoenberg1942} or more generally on compact groups \cite{erbfilbir2008}. Also they our known for symmetric spaces as the unit sphere. Here, the corresponding positive definite functions providing the kernels are referred to as spherical basis functions (SBF's) \cite{Hubbert2005,Mhaskar1999}.  

Positive definite functions and their generalizations have a long and rich mathematical history, and they are corner elements in harmonic analysis, signal processing and probability theory \cite{Sasvari1994,Stewart1976}. The core characterization of positive definite functions is given by Bochner's Theorem \cite{Bochner1933} linking positive definiteness in space to positivity in the Fourier domain. This link is essential for a lot of applications. Most prominently, in signal processing the standard kernels for convolution and signal filtering are all based upon positive definite functions, as the sinc filter or the Gaussian filter. Further, for RBF's and SBF's, the positive definiteness of the basis function guarantees the positive definiteness of the corresponding kernel $K$ and, thus, the unisolvence (I1) of the interpolation problem.

In the emerging field of graph signal processing, a general theory of positive definite functions has not been studied so far. The goal of this work is to introduce and promote the concept of positive definite functions on graphs and to investigate their role as generators of kernel-based interpolation schemes. In analogy to RBF interpolation in $\Rr^d$, we want to understand how interpolation with generalized shifts of a positive definite graph basis function (GBF) can be implemented on graphs. Further, based on the Fourier properties of the positive definite basis function, we aim to give more details about the approximation power and the stability properties of the GBF interpolation scheme. 

\vspace{2mm}   

\noindent \textbf{Our main contributions are:} 
\begin{itemize}
\item We give a proper definition of positive definite functions in spectral graph theory embedded in the structure of a graph $C^*$-algebra.  
Further, we figure out several characteristic properties of positive definite functions: we give a Bochner-type characterization based on the graph Fourier transform and a description in terms of moment conditions on the graph (Section \ref{sec:pdfunctions}).
\item We describe how generalized translates of positive definite graph basis functions can be applied as a kernel-based interpolation scheme on graphs. We show how the native  spaces for the interpolation are characterized in terms of the positive definite GBF's (Section \ref{sec:GBFinterpolation}).
\item We analyze space-frequency decompositions on graphs that are based upon positive definite window functions (Section \ref{sec:spacefrequency}).    
\item We study stability properties of the GBF interpolation scheme and provide estimates for the approximation error (Section \ref{sec:errorbounds}).
\item We show how the GBF interpolation scheme can be used to obtain quadrature rules for the integration of graph signals and derive bounds for the respective errors (Section \ref{sec:integration}).  
\end{itemize}
Required terminologies and preliminary results for spectral graph theory and kernel-based interpolation are derived in Section \ref{sec:spectraltheory} and Section \ref{sec:kernelmethods}. Further, in Section \ref{sec:examples} we provide a list of useful GBF's.

\section{Background} \label{sec:spectraltheory}
\subsection{Introduction to spectral graph theory}
We start this work with a general overview on spectral graph theory and
the notions of graph Fourier transform, graph spectrum and graph convolution. 
A standard reference for spectral graph theory is the monography \cite{Chung} by F. Chung. For an introduction to the graph Fourier transform, graph convolution and space-frequency concepts related to graphs, we refer to \cite{shuman2016}. 

In our considerations, the graph $G$ is a triplet $G=(V,E,\mathbf{A})$, where 
$V=\{\node{v}_1, \ldots, \node{v}_{n}\}$ denotes a finite set of vertices, $E \subseteq V \times V$ is the set of (directed or undirected) edges connecting the vertices and $\mathbf{A} \in \mathbb{R}^{n \times n}$ is a weighted, symmetric
and non-negative adjacency matrix containing the connection weights of the edges. The harmonic structure of the graph $G$ is encoded in this adjacency matrix $\mathbf{A}$. As $\mathbf{A}$ is assumed to be symmetric the harmonic structure on $G$ is inherently undirected, also if the edges of $G$ are directed. 

Goal of this work is to study interpolation of signals $x: V \rightarrow \mathbb{R}$ on the graph $G$. We denote the vector space of all
signals on $G$ as $\mathcal{L}(G)$. Since the number of nodes in $G$ is fixed, the dimension of $\mathcal{L}(G)$ is finite and corresponds to the number $n$ of nodes. As the node set $V$ is ordered, we can describe the signal $x$ also as a vector 
$x = (x(\node{v}_1), \ldots, x(\node{v}_n))^{\intercal}\in \mathbb{R}^n$. Depending on the context, we will switch between the representation of $x$ as a function in $\mathcal{L}(G)$ and a vector in $\Rr^n$. On the space $\mathcal{L}(G)$, we have a natural inner product given by $$y^\intercal x := \sum_{i=1}^n x(\node{v}_i) y(\node{v}_i).$$ The corresponding euclidean norm is given by $\|x\|^2 := x^{\intercal} x = \sum_{i=1}^n x(\node{v}_i)^2$. The canonical orthonormal basis in $\mathcal{L}(G)$ is denoted by $\{e_1, \ldots, e_n\}$ and given by the unit vectors $e_j$ satisfying
$e_j(\node{v}_i) = \delta_{ij}$ for $i,j \in \{1, \ldots,n\}$. 

We consider the (normalized) graph Laplacian $\mathbf{L}$ associated to the adjacency matrix $\mathbf{A}$ to determine a spectral structure on $G$:
\begin{equation*}
    \mathbf{L} := \mathbf{I_n} - \mathbf{D}^{-\frac{1}{2}}\mathbf{A}\mathbf{D}^{-\frac{1}{2}}.
\end{equation*}
Here, $\mathbf{I_n}$ denotes the identity operator in $\Rr^n$, and $\mathbf{D}$ is the degree matrix with entries given as
\begin{equation*}
    \mathbf{D}_{ij} := \left.
  \begin{cases}
    \sum_{k=0}^n \mathbf{A}_{ik}, & \text{if } i=j \\
    0, & \text{otherwise}
  \end{cases}.
  \right.
\end{equation*}
As $\mathbf{A}$ is symmetric, also the graph Laplacian $\mathbf{L}$ is a symmetric matrix and we can compute its orthonormal eigendecomposition as
\begin{equation*}
\mathbf{L}=\mathbf{U}\mathbf{M}_{\lambda} \mathbf{U^\intercal},
\end{equation*}
where 
$\mathbf{M}_{\lambda} = \mathrm{diag}(\lambda) = \text{diag}(\lambda_1,\ldots,\lambda_{n})$ 
is the diagonal matrix with the increasingly ordered eigenvalues $\lambda_i$, $i \in \{1, \ldots, n\}$, of $\mathbf{L}$ as diagonal entries.
The columns $ u_1, \ldots, u_{n}$ of the orthonormal matrix $\mathbf{U}$ are normalized eigenvectors of
$\mathbf{L}$ with respect to the eigenvalues $\lambda_1, \ldots, \lambda_n$. The ordered set $\hat{G} = \{u_1, \ldots, u_{n}\}$ of eigenvectors is an orthonormal basis for the space of signals on the graph $G$. We call $\hat{G}$ the spectrum of the graph $G$.

\subsection{Fourier transform on graphs} 
In classical Fourier analysis, as for instance the Euclidean space or the torus, the Fourier transform can be defined in terms of the eigenvalues and eigenfunctions of the Laplace operator.
In analogy, we consider the elements of $\hat{G}$, i.e. the eigenvectors $\{u_1, \ldots, u_{n}\}$, as the Fourier basis on the graph $G$. In particular, going back to our 
spatial signal $x$, we can define the graph Fourier transform of $x$ as
\begin{equation*}
\hat{x} := \mathbf{U^\intercal}x  = (u_1^\intercal x, \ldots, u_n^\intercal x)^{\intercal},
\end{equation*}
and its inverse graph Fourier transform as
\begin{equation*}
x := \mathbf{U}\hat{x}. 
\end{equation*}
The entries $\hat{x}_i = u_i^\intercal x$ of $\hat{x}$ are the frequency components or coefficients
of the signal $x$ with respect to the basis function $u_i$. For this reason, $\hat{x} : \hat{G} \to \Rr$ can be regarded as a function on the spectral domain $\hat{G}$ of the graph $G$. 
To keep the notation simple, we will however usually represent spectral distributions $\hat{x}$ as 
vectors $(\hat{x}_1, \ldots, \hat{x}_n)^{\intercal}$ in $\Rr^n$. Regarding the eigenvalues of the normalized
graph Laplacian $\mathbf{L}$, it is well-known that (see \cite[Lemma 1.7]{Chung})
\[0 = \lambda_1 \leq \lambda_2 \leq \cdots \leq \lambda_n \leq 2.\]
Note that it is possible to use the spectral decomposition of other suitable operators on $\mathcal{L}(G)$ instead of the normalized graph Laplacian $\mathbf{L}$ in order to define the graph Fourier transform $\mathbf{U}$ on $G$. Common examples in the literature include the adjacency matrix $\mathbf{A}$ or other normalizations of $\mathbf{L}$. All the results of this work hold true also for these alternative generators of the graph Fourier transform as long as the Fourier matrix $\mathbf{U}$ is orthogonal.   

\subsection{Convolution on graphs}

With the graph Fourier transform we obtain the possibility to define a convolution between two graph signals $x$ and $y$. In analogy to classical Fourier analysis in which the convolution of two signals is calculated as the pointwise product of their Fourier transforms, we define for $x,y \in \mathcal{L}(G)$ the graph convolution as
\begin{equation}
    x \ast y := \mathbf{U} \left ( \mathbf{M}_{\hat{x}} \hat{y} \right ) = \mathbf{U}\mathbf{M}_{\hat{x}}\mathbf{U^\intercal}y \label{eq:spectralfilter1}.
\end{equation}
As before, $\mathbf{M}_{\hat{x}}$ denotes the diagonal matrix 
$\mathbf{M}_{\hat{x}} = \mathrm{diag}(\hat{x})$ and $\mathbf{M}_{\hat{x}} \hat{y} = (\hat{x}_1 \hat{y}_1, \ldots, \hat{x}_{n} \hat{y}_{n})$
gives the pointwise product of the two vectors $\hat{x}$ and $\hat{y}$. 
The convolution $\ast$ on $\mathcal{L}(G)$ has the following properties:
\begin{itemize}
\item[(1)] $x \ast y = y \ast x$ (Commutativity),
\item[(2)] $(x \ast y) \ast z = x \ast (y \ast z)$ (Associativity),
\item[(3)] $(x + y) \ast z = x \ast z + (y \ast z)$ (Distributivity),
\item[(4)] $(\alpha x) \ast y = \alpha(y \ast x)$ for all $\alpha \in \Rr$ (Associativity for scalar multiplication).
\end{itemize}
The unity element of the convolution is given by $\funit = \sum_{i=1}^n u_i$. 
In view of the linear structure in \eqref{eq:spectralfilter1}, we can further define a convolution operator $\mathbf{C}_x$ on $\mathcal{L}(G)$ as
\[\mathbf{C}_x = \mathbf{U}\mathbf{M}_{\hat{x}}\mathbf{U^\intercal}. \]
The convolution $x * y$ can then be formulated as the matrix-vector product $\mathbf{C}_x y = x \ast y$. Written in this way, we can regard every signal $x \in \mathcal{L}(G)$ also as a filter function acting by convolution on a second signal $y$. 

\subsection{The graph $C^{\ast}$-algebra}

The rules (1)-(4) of the graph convolution guarantee that the vector space $\mathcal{L}(G)$ endowed with the convolution $\ast$ as a multiplicative operation is a commutative and associative algebra. With the identity as a trivial involution and the norm 
$$\|x\|_{\mathcal{A}} = \sup_{\|y\| = 1} \|x \ast y\|$$
we obtain a real $C^{\ast}$-algebra $\mathcal{A}$. Relevant for us is the fact that $\mathcal{A}$ is commutative and finite. A general introduction 
to $C^{\ast}$-algebras with the description of commutative and finite $C^{\ast}$-algebras can be found in \cite{Davidson1996}. Further characterizations of real $C^{\ast}$-algebras are, for instance, given in \cite{Li2003}. 

The graph $C^{\ast}$-algebra $\mathcal{A}$ is the standard model for graph signal processing in this work. The algebra $\mathcal{A}$ contains all possible signals and filter functions on $G$ and describes how filters act on signals via convolution. Furthermore, $\mathcal{A}$ contains the entire information of the graph Fourier transform. This can be seen as follows: the spectrum of the commutative $C^{\ast}$-algebra $\mathcal{A}$ is given by the set of multiplicative linear functionals on $\mathcal{A}$ that preserve the multiplicative structure of the algebra. These so-called characters of $\mathcal{A}$ are in our case exactly the $n$ functionals $x \to u_k^{\intercal} x$, $k \in \{1,\ldots,n\}$. The spectrum of the $C^{\ast}$-algebra $\mathcal{A}$ can therefore be naturally identified with the already introduced spectrum $\hat{G} = \{u_1, \ldots, u_n\}$ of the graph $G$. The famous Gelfand-Naimark theorem translated to the $C^{\ast}$-algebra $\mathcal{A}$ then confirms that the graph Fourier transform $\mathbf{U}^\intercal$ is an algebra isomorphism between the $C^{\ast}$-algebra $\mathcal{A}$ and the $C^{\ast}$-algebra of functions on the spectrum $\hat{G}$ with the pointwise multiplication as multiplicative operation. Our study of positive definite functions on the graph $G$ essentially relies on this algebraic signal model.  

The algebraic structure of graph signal processing was observed in several previous works. In \cite{PuschelMoura2008}, a general algebraic signal model was developed in order to describe signal processing in discrete settings. The $C^{\ast}$-algebra framework considered in this work fits as a particular case in the general framework of \cite{PuschelMoura2008}, offers however a much closer relation between signals, convolution and graph Fourier transform.

Beside the $C^{\ast}$-algebra $\mathcal{A}$ we will consider the following two subalgebras:
\begin{itemize}
\item[(1)] The $C^{\ast}$-algebra $\mathcal{A}_{\mathbf{L}}$ generated by the graph Laplacian $\mathbf{L}$ as
$$\mathcal{A}_{\mathbf{L}} := \mathrm{span} \{\funit, \mathbf{L} \funit, 
\mathbf{L}^2 \funit, \ldots,
\mathbf{L}^{n-1} \funit\}. $$
$\mathcal{A}_{\mathbf{L}}$ is a subalgebra of $\mathcal{A}$ that contains the unity element $\funit$ of the convolution. The algebra $\mathcal{A}_{\mathbf{L}}$ is relevant for the construction of filter functions in terms of the Laplacian $\mathbf{L}$. 
\item[(2)] The $C^{\ast}$-algebra $\mathcal{B}_{M}$ of bandlimited signals with bandwidth $M \leq n$ given by
\[\mathcal{B}_{M} := \mathrm{span} \{u_1, \ldots, u_M\}. \]
If $M < n$, then $\funit$ is not contained in $\mathcal{B}_M$. The multiplicative unity of the subalgebra $\mathcal{B}_{M}$ is in this case given by $\sum_{k=1}^M u_k$. The bandlimited signals are relevant for us as approximation spaces.
\end{itemize}

We conclude this section with a characterization of the subalgebra $\mathcal{A}_{\mathbf{L}}$. 

\begin{proposition} \label{prop-AL}
Assume that the graph Laplacian $\mathbf{L}$ has precisely $r \leq n$ distinct eigenvalues. 
Then, $\mathcal{A}_{\mathbf{L}}$ is a $r$-dimensional subalgebra of $\mathcal{A}$. A signal $x$ is contained in $\mathcal{A}_{\mathbf{L}}$ if and only if
$\hat{x}_k = \hat{x}_{k'}$ whenever $\lambda_k = \lambda_{k'}$. Furthermore,
$$\mathcal{A}_{\mathbf{L}} = \{\funit, \mathbf{L} \funit, \ldots \mathbf{L}^{r-1} \funit\}.$$  
\end{proposition}

\begin{proof}
We consider the algebra
\[\mathcal{\tilde{A}} = \{x \in \mathcal{A} \ | \ \hat{x}_k = \hat{x}_{k'} \ \text{if} \ \lambda_k = \lambda_{k'}\}\]
and show that $\mathcal{\tilde{A}}$ corresponds to $\mathcal{A}_{\mathbf{L}}$.
Clearly, $\mathcal{\tilde{A}}$ is a subalgebra of $\mathcal{A}$ that contains the unity $\funit$ of the convolution. As $\mathbf{U}^{\intercal} \mathbf{L} x = \mathbf{M}_{\lambda} \mathbf{U}^{\intercal} x$, we see that $\mathbf{L} x \in \mathcal{\tilde{A}}$ if $x \in \mathcal{\tilde{A}}$. This implies that $\mathcal{A}_{\mathbf{L}}$ is contained in $\mathcal{\tilde{A}}$. As the dimension of $\mathcal{\tilde{A}}$ corresponds to the number $r$ of distinct eigenvalues of $\mathbf{L}$, it only remains to show that $\{\funit, \mathbf{L} \funit, \ldots \mathbf{L}^{r-1} \funit\} \subset \mathcal{A}_{\mathbf{L}}$ is a system of $r$ linear independent vector space elements. To see this, we pick $r$ distinct eigenvalues of the graph Laplacian $\mathbf{L}$ and denote them by $\lambda_{k_1}, \ldots, \lambda_{k_r}$. 
Then, we have $u_{k_j}^{\intercal} \mathbf{L}^{i-1} \funit = \lambda_{k_j}^{i-1}$ for $i,j \in \{1, \ldots, r\}$. 
Now, as the Vandermonde matrix 
\begin{equation*} \mathbf{V}_r = 
\begin{pmatrix} \lambda_{k_1}^0 & \lambda_{k_1}^1 & \ldots & \lambda_{k_1}^{r-1} \\
\vdots & \vdots & \ddots & \vdots \\
                \lambda_{k_r}^0 & \lambda_{k_r}^1 & \ldots & \lambda_{k_r}^{r-1} 
\end{pmatrix}\end{equation*}
is invertible, the matrix $(\funit, \mathbf{L} \funit, \ldots, \mathbf{L}^{r-1} \funit)$
has full rank and its columns are linearly independent. 
\end{proof}

Proposition \ref{prop-AL} states that $\mathcal{A} = \mathcal{A}_{\mathbf{L}}$ holds true if and only if the spectrum of the graph Laplacian $\mathbf{L}$ is simple, that is, if all the eigenvalues of $\mathbf{L}$ are distinct. In this case, the dimension of the algebra $\mathcal{A}_{\mathbf{L}}$ is largest possible. For some graphs the subalgebra $\mathcal{A}_{\mathbf{L}}$ can however also be particularly small. For example, if $G$ is an unweighted complete graph, i.e. every pair of nodes in $V$ is connected by a unique, equally weighted edge, then the normalized graph Laplacian possesses only the eigenvalues $\lambda_1 = 0$ and $\lambda_2 = \lambda_3 = \cdots = \lambda_n = 2$. In this case, $\mathcal{A}_{\mathbf{L}} = \mathrm{span} \{\funit, \mathbf{L} \funit\}$ is only two-dimensional.  

\section{Kernel-based methods for interpolation on graphs} \label{sec:kernelmethods}

In this section, we give a synthesis of well-known facts about kernel-based interpolation methods on discrete sets. In the classical euclidean setting, the corresponding concepts are, for instance, presented in \cite{Scha98,SchabackWendland2003} or in the treatise \cite{We05}. An introduction to kernel-based methods for machine learning is given in \cite{Schoelkopf2002}.  
A recent survey on the history and research trends related to positive definite kernels can be found in \cite{Fasshauer2011}. 
Note that the following derivations do not yet take into account spectral or geometric information of the graph $G$.

\subsection{Positive definite kernels on graphs}
We are interested in kernel functions $K : V \times V \to \Rr$ on the graph $G$ that are symmetric, i.e., they satisfy $K(\node{v},\node{w}) = K(\node{w},\node{v})$ for all nodes $\node{v},\node{w} \in V$. A kernel $K$ allows to introduce a linear operator $\mathbf{K}: \mathcal{L}(G) \to \mathcal{L}(G)$ acting on a graph signal $x \in\mathcal{L}(G)$ as
\[\mathbf{K} x(\node{v}_i) = \sum_{j=1}^n K(\node{v}_i,\node{v}_j) x(\node{v}_j).\] 
Based on our identification of signals $x \in \mathcal{L}(G)$ with vectors in 
$\Rr^n$, we can represent $\mathbf{K}$ as the symmetric matrix
$\mathbf{K} \in \Rr^{n \times n}$ given by
\[ \mathbf{K} = \begin{pmatrix} K(\node{v}_1,\node{v}_1) & K(\node{v}_1,\node{v}_2) & \ldots & K(\node{v}_1,\node{v}_n) \\
K(\node{v}_2,\node{v}_1) & K(\node{v}_2,\node{v}_2) & \ldots & K(\node{v}_2,\node{v}_n) \\
\vdots & \vdots & \ddots & \vdots \\
K(\node{v}_n,\node{v}_1) & K(\node{v}_n,\node{v}_2) & \ldots & K(\node{v}_n,\node{v}_n)
\end{pmatrix}.\]
The following families of symmetric kernels are particularly relevant for this work:
\begin{definition} \label{def:pd}
\leavevmode
\begin{itemize}
\item[(1)] We call a symmetric kernel $K$ positive semi-definite (p.s.d.) 
if the matrix $\mathbf{K} \in \Rr^{n \times n}$ is positive semi-definite, i.e., $x^{\intercal} \mathbf{K} x \geq 0$ for all $x \in \Rr^n$. 
\item[(2)] We call a symmetric kernel $K$ positive definite (p.d.) if the matrix $\mathbf{K} \in \Rr^{n \times n}$ is strictly positive definite, i.e., we have $x^{\intercal} \mathbf{K} x > 0$ for all $x \in \Rr^n$, $x \neq 0$.
\item[(3)] We call $K$ conditionally positive definite (c.p.d.) with respect to a subspace $\mathcal{Y} \subset \mathcal{L}(G)$, if $\mathbf{K}$ is p.d. on the subspace $\mathcal{Y}$. 
\end{itemize}
\end{definition}
\subsection{Interpolation with positive definite kernels} \label{subsec:kernelinterpolation}
Every p.d. kernel $K$ allows to equip the space
$\mathcal{L}(G)$ with an inner product of the form 
\[ \langle x,y \rangle_{K} = y^{\intercal} \mathbf{K}^{-1} x, \qquad x,y \in \mathcal{L}(G).\]
The resulting inner product space, referred to as native space $\mathcal{N}_{K}$,
is a reproducing kernel Hilbert space \cite{Aronszajn1950} in which $K$ assumes the role of the reproducing kernel satisfying the property
\[ \langle x, K(\cdot,\node{v}_j) \rangle_{K} = x^{\intercal} \mathbf{K}^{-1} K(\cdot,\node{v}_j) = x(\node{v}_j) \quad \text{for all $x \in \mathcal{L}(G)$}.\]

A p.d. kernel $K$ can generally be used to solve interpolation problems in  interpolation spaces that are generated by columns of the matrix $\mathbf{K}$. The corresponding interpolation problem on the graph $G$ reads as follows: 
for given samples $x(\node{w}_1), \ldots x(\node{w}_N)$ of a signal $x$ on a subset $W = \{\node{w}_1, \ldots, \node{w}_N\} \subset V$, $N \leq n$, find an interpolating signal $\itpx \in \mathcal{L}(G)$ that interpolates $x$ at the nodes in $W$, i.e.,
\begin{equation} \label{eq:interpolation} \itpx(\node{w}_k) = x(\node{w}_k) \quad \text{for all} \; k \in \{1, \ldots, N\}.
\end{equation} 
With a p.d. kernel $K$ this interpolation problem can be solved as follows. As an interpolation basis we consider the columns $K(\cdot,\node{w}_k)$, $k \in \{1, \ldots, N\}$, of the matrix $\mathbf{K}$. Then an interpolating signal $\itpx$ based on the 
expansion 
$$\itpx(\node{v}) = \sum_{k=1}^N c_k K(\node{v},\node{w}_k)$$
has to satisfy the interpolation condition
\begin{equation} \label{eq:computationcoefficients} 
\underbrace{\begin{pmatrix} K(\node{w}_1,\node{w}_1) & K(\node{w}_1,\node{w}_2) & \ldots & K(\node{w}_1,\node{w}_N) \\
K(\node{w}_2,\node{w}_1) & K(\node{w}_2,\node{w}_2) & \ldots & K(\node{w}_2,\node{w}_N) \\
\vdots & \vdots & \ddots & \vdots \\
K(\node{w}_N,\node{w}_1) & K(\node{w}_N,\node{w}_2) & \ldots & K(\node{w}_N,\node{w}_N)
\end{pmatrix}}_{\mathbf{K}_W} \begin{pmatrix} c_1 \\ c_2 \\ \vdots \\ c_N \end{pmatrix}
= \begin{pmatrix} x(\node{w}_1) \\ x(\node{w}_2) \\ \vdots \\ x(\node{w}_N) \end{pmatrix}.
\end{equation}
As $\mathbf{K}$ is p.d. also the submatrix $\mathbf{K}_W$ is p.d. by the inclusion principle \cite[Theorem 4.3.15]{HornJohnson1985}. The linear system \eqref{eq:computationcoefficients} therefore has a unique solution, and the interpolating signal $\itpx$ can be written uniquely in terms of the basis $\{K(\cdot,\node{w}_1), \ldots, K(\cdot,\node{w}_N)\}$. 
We denote the corresponding interpolation space as
\[\mathcal{N}_{K,W} = \left\{x \in \mathcal{L}(G) \ | \ x(\node{v}) = \sum_{k=1}^N c_k K(\node{v},\node{w}_k)\right\}. \]

The following result is standard for reproducing kernel Hilbert spaces and one of the reasons why these spaces are so popular for the interpolation and approximation of signals. As the proof is almost a one-liner, we add it at this place.  

\begin{proposition} (\cite[Corollary 10.25]{We05}) \label{prop:minimumenergy}
The interpolant $\itpx$ of $x$ minimizes the native space norm $\|\cdot\|_{K}$ over all other possible interpolants of $x$ in $\mathcal{L}(G)$ at the nodes $W$. 
\end{proposition}

\begin{proof}
If a signal $y$ vanishes on $W$, the reproducing property of the kernel $K$ yields the identity
\[\langle y, \itpx\rangle_K = \langle y, \sum_{k=1}^N c_k K(\cdot,\node{w}_k) \rangle_K
= \sum_{k=1}^N c_k y(\node{w}_k) = 0. \]
Therefore, if $z \in \mathcal{L}(G)$ is a second interpolant of $x$ at the nodes $W$, we get
\[\|\itpx\|_K^2 = \langle \itpx, \itpx - z + z \rangle_K = 
\langle \itpx, z \rangle_K \leq \|\itpx\|_K \|z\|_K.\]
\end{proof}

\subsection{Interpolation with conditionally positive definite kernels} 
\label{subsec:cpdkernels}
If $K$ is c.p.d. with respect to a subspace $\mathcal{Y}$, the interpolation on the node set $W$ can be performed in a similar way once the issue with the non positive definiteness on the orthogonal complement $\mathcal{Y}^{\perp}$ of $\mathcal{Y}$ is solved. If the dimension $M$ of the complement $\mathcal{Y}^\perp$ is small and an orthonormal basis $\{y_1^\perp, \ldots y_M^{\perp}\}$ of $\mathcal{Y}^\perp$ is given, this issue can be fixed in a simple manner by defining the augmented kernel
\begin{equation} \label{eq:augmentedkernel} K^{(\delta)}(\node{v},\node{w}) = K(\node{v},\node{w}) + \delta \left( \sum_{i=1}^M y_i^\perp(\node{v}) y_i^\perp(\node{w})\right).
\end{equation}
If the parameter $\delta > |\lambda_{\min}(\mathbf{K})| \geq 0$ is larger 
than the modulus of the smallest eigenvalue of $\mathbf{K}$, then the augmented
kernel $K^{(\delta)}$ is positive definite and we can apply the interpolation procedure of the last section. In particular we can find a unique interpolation signal 
$x \in \mathcal{N}_{K^{(\delta)},W}$ such that the interpolation problem \eqref{eq:interpolation} is solved. The interpolant $\itpx$ has the  
expansion 
$$\itpx(\node{v}) = \sum_{k=1}^N c_k K(\node{v},\node{w}_k) + \sum_{i=1}^M d_i y_i^{\perp}(\node{v}), \qquad d_i  = \delta \sum_{k=1}^N c_k y_k^{\perp}(\node{w}_k). $$
where the coefficients $c_k$ are the solutions of \eqref{eq:computationcoefficients}
with respect to the augmented kernel $K^{(\delta)}$. In particular, the interpolation space $\mathcal{N}_{K^{(\delta)},W}$ is a $N$-dimensional subspace of the space $\mathcal{N}_{K,W} + \mathcal{Y}^\perp$.

\begin{remark}
The sum $\sum_{i=1}^M y_i^\perp(\node{v}) y_i^\perp(\node{w})$ in \eqref{eq:augmentedkernel} can be regarded as the reproducing kernel of the 
orthogonal complement $\mathcal{Y}^\perp$. By adding a $\delta$-multiple of this kernel to the c.p.d. kernel $K$, the non-positive part of the spectrum of $K$ is shifted by $\delta>0$ in positive direction, resulting in a p.d. kernel $K^{(\delta)}$. 
An alternative strategy to obtain a p.d. kernel from a c.p.d. kernel is formulated in \cite{BerschneiderCastell2009} in terms of a reflection technique and Pontryagin spaces. A third possibility (for instance pursued in \cite{Pesenson2009}) consists in adding a multiple of the identity matrix to the c.p.d. kernel $K$ and, thus, in shifting the entire spectrum of the kernel to the positive real axis.       
\end{remark} 

\begin{remark}
Definition \ref{def:pd} (3) for c.p.d. kernels is not the standard definition given in the literature, see for instance, \cite{We05}. The standard definition reads as follows: $K$ is c.p.d. if and only if for all subsets $W$ the matrix
$\mathbf{K}_W$ is p.d. on the subspace determined by $\sum_{k=1}^N y^{\perp}(\node{w}_k) c_k = 0$ for all $y^{\perp} \in \mathcal{Y}^{\perp}$. In \cite{BerschneiderCastell2009} it is shown, that every kernel that is c.p.d. with respect to this standard definition can be interpreted as a kernel that is c.p.d. with respect to Definition \ref{def:pd} (3) and vice versa. In some works, the c.p.d. kernels in Definition \ref{def:pd} (3) are referred to as kernels with finitely many negative squares, see \cite{BerschneiderCastell2009,Sasvari1994}.
\end{remark}
  
\section{Positive definite functions on graphs} \label{sec:pdfunctions}
\subsection{Definition}
The general kernel-based interpolation scheme of the last section does not include spectral information of the graph. Goal of this section is to harmonize these two structures and to develop an interpolation scheme in which the interpolation kernels are characterized in terms of the generalized translates of a single graph basis function. For this, the following notion of positive definiteness is essential.

\begin{definition} \label{def:pdfunction}
We call a function $f: V \to \Rr$ on the graph $G$ positive semi-definite (positive definite) if the matrix 
\[ \mathbf{K}_{f} = \begin{pmatrix} \mathbf{C}_{e_{1}} f(\node{v}_1) & \mathbf{C}_{e_{2}} f(\node{v}_1) & \ldots & \mathbf{C}_{e_{n}} f(\node{v}_1) \\
\mathbf{C}_{e_{1}} f(\node{v}_2) & \mathbf{C}_{e_{2}} f(\node{v}_2) & \ldots & \mathbf{C}_{e_{n}} f(\node{v}_2) \\
\vdots & \vdots & \ddots & \vdots \\
\mathbf{C}_{e_{1}} f(\node{v}_n) & \mathbf{C}_{e_{2}} f(\node{v}_n) & \ldots & \mathbf{C}_{e_{n}} f(\node{v}_n)
\end{pmatrix}\]
is symmetric and positive semi-definite (positive definite, respectively). We call $f$ conditionally positive definite with respect to a subspace $\mathcal{Y}$ if $\mathbf{K}_{f}$ is p.d. on $\mathcal{Y}$. 
The sets of positive semi-definite and positive definite functions in $\mathcal{L}(G)$ are denoted by $\mathcal{P}$ and $\mathcal{P}_+$, respectively.  
\end{definition}
A positive semi-definite function $f$ induces naturally a p.s.d kernel $K_f$ on $G$ by 
\[ K_f(\node{v}_i,\node{v}_j) := \mathbf{C}_{e_{j}} f(\node{v}_i).\]
If there is an additional group structure on $G$, the signal $\mathbf{C}_{e_{i}} f$ corresponds precisely to the shift of the signal $f$ by the group element $\node{v}_i$. On general graphs, we don't have an inherent notion of translation. Nevertheless, we will encounter several examples in which the signals $\mathbf{C}_{e_{i}} f$ are spatially well-localized around the nodes $\node{v}_i$. For this reason, we can interpret $\mathbf{C}_{e_{i}} f$ as a generalized translate of $f$ on $G$. In any case, the positive definiteness of $f$ implies that the set $\{ \mathbf{C}_{e_{1}} f, \ldots, \mathbf{C}_{e_{n}} f \}$ of generalized shifts of $f$ is linearly independent and forms a basis of $\mathcal{L}(G)$.   

\subsection{A Bochner-type characterization of positive definite functions}
The introduced notion of positive definiteness is deeply linked to the 
spectrum $\hat{G} = \{u_1, \ldots, u_n\}$ of the graph $G$. A direct manifestation of this link is the following Bochner-type characterization of a p.d. function $f$ in terms of the graph Fourier transform $\hat{f}$.

\begin{theorem} \label{thm:Bochner}
A function $f \in \mathcal{L}(G)$ is contained in $\mathcal{P}$ (in $\mathcal{P}_+$) if and only if $\hat{f}_k \geq 0$ ($\hat{f}_k > 0$, respectively) for all $k \in \{1, \ldots, n\}$. The corresponding p.s.d. kernel $K_f$ has the Mercer decomposition
\[ K_f(\node{v},\node{w}) = \sum_{k=1}^n \hat{f}_k \, u_k(\node{v}) \, u_k(\node{w}).\]
Further, we have the following refinements:
\begin{enumerate}
\item[(i)] $f \in \mathcal{A}_{\mathbf{L}} \cap \mathcal{P}$ if and only if $\hat{f}_k \geq 0$ and $\hat{f}_k = \hat{f}_{k'}$ for $\lambda_{k} = \lambda_{k'}$.
\item[(ii)] $f \in \mathcal{B}_M \cap \mathcal{P}$ if and only if $\hat{f}_k \geq 0$ and $\hat{f}_k = 0$ for all indices $k > M$.  
\item[(iii)] $f$ is c.p.d. with respect to the subspace $\mathcal{Y} = \mathrm{span} \{u_{k_1}, \ldots, u_{k_K}\}$ if and only if $\hat{f}_{k_1} > 0, \ldots, \hat{f}_{k_K} > 0$.
\end{enumerate} 
\end{theorem}

\begin{proof}
We apply the definition of the convolution given in \eqref{eq:spectralfilter1} to the kernel matrix $\mathbf{K}_f$. With the convolution operator given by $\mathbf{C}_x = \mathbf{U}\mathbf{M}_{\hat{x}}\mathbf{U^\intercal}$ and the commutativity of the convolution, we can rewrite the columns of the kernel matrix $\mathbf{K}_f$ as 
\[K_f(\cdot,\node{v}_i) = \mathbf{C}_{e_i} f = \mathbf{C}_{f} e_i = \mathbf{U}\mathbf{M}_{\hat{f}}\mathbf{U^\intercal} e_i.\]
In this way, $\mathbf{K}_f = \mathbf{U}\mathbf{M}_{\hat{f}}\mathbf{U^\intercal}$, and we have found the spectral decomposition of the matrix $\mathbf{K}_f$ as well as the Mercer decomposition of $K_f$. In particular, the entries $\hat{f}_k$ of $\hat{f}$ are precisely the eigenvalues of $\mathbf{K}_f$. This implies that $\mathbf{K}_f$ is positive semi-definite (p.d.) if and only if $\hat{f}_k \geq 0$ ($\hat{f}_k > 0$) for
all $k \in \{1,\ldots,n\}$. The additional refinements for the subalgebras $\mathcal{A}_{\mathbf{L}}$ and $\mathcal{B}_M$ as well as for c.p.d. functions follow directly by the respective definitions and Proposition \ref{prop-AL}.
\end{proof}

Note that, while Bochner's characterization of p.s.d. functions in $\Rr^d$ is a rather deep result, Theorem \ref{thm:Bochner} is an almost direct consequence of Definition \ref{def:pdfunction} for p.d. functions on graphs. The reason for this is twofold: the vector space of signals $\mathcal{L}(G)$ is only finite dimensional, simplifying many considerations regarding the involved function spaces; and the definition of the convolution in \eqref{eq:spectralfilter1} is closely linked to the spectrum of the graph $G$. 

Theorem \ref{thm:Bochner} implies that there is a one to one correlation between p.s.d. functions and p.s.d. kernels on graphs with a Mercer extension in terms of the Fourier basis $\{u_1, \ldots, u_n\}$. A second important consequence is related to the graph $C^*$-algebra $\mathcal{A}$. Theorem \ref{thm:Bochner} states that the set of p.d. functions corresponds precisely to the set of positive elements in the $C^*$-algebra $\mathcal{A}$. 

\subsection{The convex cone of positive definite functions on graphs}

We are interested in characterizing the set $\mathcal{P}$ as well as the subset $\mathcal{P}_+$ of p.d. functions. 
For this, we first consider the norm
\[ \|x\|_{\mathcal{A}'} = \sum_{k=1}^n |\hat{x}_k|,\]
and the bounded subset
\[ \mathcal{P}_1 = \{f \in \mathcal{P} \ | \ \|f\|_{\mathcal{A}'} \leq 1 \}.\]
The norm $\|\cdot\|_{\mathcal{A}'}$ indicates the norm dual to the $C^\ast$-algebra norm $\|\cdot\|_{\mathcal{A}}$. The set $\mathcal{P}_1$ therefore
corresponds to the intersection of the unit ball in the dual algebra $\mathcal{A}'$ with the set $\mathcal{P}$ of p.s.d. functions. In the following we denote the zero signal
in $\mathcal{L}(G)$ by $\zero = (0, \ldots, 0)^{\intercal}$.
As a first result, we get the following characterization of the convex set $\mathcal{P}_1$. 

\begin{theorem} \label{thm:conv}
The signals $\{\zero,u_1, \ldots, u_n\}$ are the extreme points of the convex and compact set
$\mathcal{P}_1$. In particular, we have
\[ \mathcal{P}_1 = \mathrm{conv} \, \{\zero,u_1, \ldots, u_n\}.\] 
\end{theorem}

\begin{proof}
We consider the euclidean space $\Rr^n$ with the standard basis $\{e_1, e_2, \ldots, e_n\}$. The vertices of the standard simplex $\Delta$ in $\Rr^n$ are the origin $\zero$ and the vectors $e_1, e_2, \ldots, e_n$. Obviously, $\Delta$ is convex, compact and has precisely the $n+1$ mentioned extremal points. Now, by the natural identification of $\mathcal{L}(\hat{G})$ with $\Rr^n$ we can regard $\Delta$ as a convex simplex in $\mathcal{L}(\hat{G})$ and apply the inverse Fourier transform $\mathbf{U}$. Then, Bochner's characterization in Theorem \ref{thm:Bochner} implies that
\[ \mathbf{U} \Delta = \mathcal{P}_1.\]
As $\mathbf{U}$ is linear and invertible, the image $\mathcal{P}_1$ of the convex set $\Delta$ is convex and compact. Further, the extremal points $e_k$ of $\Delta$ are mapped onto the extremal points $\mathbf{U} e_k = u_k$ of $\mathcal{P}_1$. This shows the statement of the theorem.  
\end{proof}

We conclude this section with a list of elementary properties of the sets $\mathcal{P}$ and $\mathcal{P}_+$.

\begin{corollary}
\leavevmode
\begin{enumerate}
\item[(1)] If $f,g \in \mathcal{P}$, then $\gamma_1 f + \gamma_2 g \in \mathcal{P}$ for all $\gamma_1, \gamma_2 \geq 0$ ($\mathcal{P}$ is a convex cone in $\mathcal{L}(G)$).
\item[(2)] The extreme rays of the cone $\mathcal{P}$  are given by $\{\gamma u_k \ | \ \gamma \geq 0\}$, $k \in \{1, \ldots, n\}$.
\item[(3)] The convex cone $\mathcal{P}_+$ is the open interior of $\mathcal{P}$.
\item[(4)] If $f,g \in \mathcal{P}$, then $f \ast g \in \mathcal{P}$ ($\mathcal{P}$ is closed under convolution).
\item[(5)] $f \in \mathcal{P}$ if and only if there exists a $g \in \mathcal{L}(G)$ with $f = g \ast g$.
\end{enumerate}
\end{corollary}

\begin{proof}
Property (1) follows directly from Definition \ref{def:pdfunction}. The statement in (2) is a consequence of Theorem \ref{thm:conv} with the additional observation that
$\mathcal{P} = \cup_{\gamma \geq 0} \gamma \mathcal{P}_1$. By Bochner's characterization in Theorem \ref{thm:Bochner}, a p.s.d. function $f$ is on the boundary of the cone $\mathcal{P}$ if and only if $\hat{f}_k = 0$ for at least one $k \in \{1,\ldots,n\}$. This on the other hand is equivalent for $f$ to be in $\mathcal{P} \setminus \mathcal{P}_+$. This shows (3). Thereby, the fact that $\mathcal{P}_+$ is a convex cone is also guaranteed by Definition \ref{def:pdfunction}. Finally, (4) and (5) follow from Theorem \ref{thm:Bochner} and the fact that the convolution of two signals is defined as the multiplication of their respective Fourier transforms. 
\end{proof}

Theorem \ref{thm:conv} further implies, that we can write every p.s.d. function $f \in \mathcal{P}_1$ as
\[f = \lambda_1 u_1 + \lambda_2 u_2 + \cdots + \lambda_n u_n, \quad \lambda_k \geq 0, \quad \sum_{k=1}^n \lambda_k \leq 1. \]
As $\mathcal{P}_+$ is the open interior of $\mathcal{P}$, every $f \in \mathcal{P}_+ \cap \mathcal{P}_1$ is then necessarily of the form
\[f = \lambda_1 u_1 + \lambda_2 u_2 + \cdots + \lambda_n u_n, \quad \lambda_k > 0, \quad \sum_{k=1}^n \lambda_k \leq 1. \]

\subsection{Moment conditions for positive definite functions}
We consider the moments $\funit^\intercal \mathbf{L}^j x, j \in \mathbb{N}_0$, of a signal $x \in \mathcal{L}(G)$. With these moments we can generate the Hankel matrices
\begin{equation} \label{eq:momentmatrix} {\mathbf{H}}_r(x) = 
\left(\begin{array}{lllll} \funit^\intercal x & \funit^\intercal \mathbf{L} x & \ldots & \funit^\intercal \mathbf{L}^{r-1} x \\
\funit^\intercal \mathbf{L} x & \funit^\intercal \mathbf{L}^2 x & \ldots & \funit^\intercal \mathbf{L}^{r} x \\
\quad \vdots & \quad \vdots & \ddots & \quad \vdots \\
\funit^\intercal \mathbf{L}^{r-1} x & \funit^\intercal \mathbf{L}^{r} x & \ldots & \funit^\intercal \mathbf{L}^{2r-2} x
\end{array} \right), \quad r \in \Nn.\end{equation}
The moment matrices ${\mathbf{H}}_r(x)$ allow to characterize positive definite functions  in the subalgebra $\mathcal{A}_{\mathbf{L}}$. 

\begin{theorem} \label{thm:moment}
Assume that the graph Laplacian $\mathbf{L}$ has exactly $r$ distinct eigenvalues. 
A signal $f \in \mathcal{A}_{\mathbf{L}}$ is positive (semi-) definite if and only if the matrix ${\mathbf{H}}_r(f)$ is positive (semi-) definite.
\end{theorem}

\begin{proof}
We first show that $f \in \mathcal{P}$ implies that the matrix ${\mathbf{H}}_r(f)$ is p.s.d.: For an arbitrary vector $y = (y_1, \ldots, y_r)^\intercal \in \Rr^r$, we have the identities
\begin{align} \label{eq:auxmoment}
y^{\intercal} {\mathbf{H}}_r(f) y = 
\sum_{i=1}^r \sum_{j = 1}^r \funit^\intercal \mathbf{L}^{i+j-2} f \, y_i y_j 
= \sum_{i=1}^r \sum_{j = 1}^r \sum_{k=1}^n 1 \hat{f}_k \lambda_k^{i+j-2} y_i y_j = \sum_{k=1}^n \hat{f}_k \big( (\lambda_k^0, \lambda_k^1, \ldots, \lambda_k^{r-1}) y \big)^2. 
\end{align}
Therefore, if $f$ is p.s.d., then Theorem \ref{thm:Bochner} implies that also ${\mathbf{H}}_r(f)$ is p.s.d. 
For the converse conclusion we need the $r$ distinct eigenvalues of the graph Laplacian $\mathbf{L}$: we denote them by $\lambda_{k_1}, \ldots, \lambda_{k_r}$. As they are distinct, the Vandermonde matrix 
\begin{equation*} \mathbf{V}_r = 
\begin{pmatrix} \lambda_{k_1}^0 & \lambda_{k_1}^1 & \ldots & \lambda_{k_1}^{r-1} \\
\vdots & \vdots & \ddots & \vdots \\
                \lambda_{k_r}^0 & \lambda_{k_r}^1 & \ldots & \lambda_{k_r}^{r-1} 
\end{pmatrix}\end{equation*}
is invertible. Thus, for every $j \in \{1, \ldots, r\}$ there exists a (unique) vector
$y^{(j)} \in \Rr^r$, $y^{(j)}\neq \zero$ such that $\mathbf{V}_r y^{(j)} = e_j \in \Rr^r$. Plugging these solutions $y^{(j)}$ into 
\eqref{eq:auxmoment}, we get 
\begin{align*}
y^{(j)\intercal} \, {\mathbf{H}}_r(f) \, y^{(j)} = 
 \sum_{k=1}^n \hat{f}_k \big( (\lambda_k^0, \lambda_k^1, \ldots, \lambda_k^{n-1}) x^{(j)} \big)^2 = \sum_{k:\lambda_k = \lambda_{k_j}} \hat{f}_k. 
\end{align*}
Therefore, if we assume that ${\mathbf{H}}_r(f)$ is p.s.d., and we use the 
characterization of the subalgebra $\mathcal{A}_{\mathbf{L}}$ in Proposition \ref{prop-AL}, we obtain $\hat{f}_k \geq 0$ for all $k \in \{1,\ldots,n\}$. This on the other hand implies that $f$ is a p.s.d. function.
For the stricter assumption that $f \in \mathcal{P}_+$, the argumentation line in the proof is almost the same. The only difference is in the demonstration of the forward direction $(f \in \mathcal{P}_+) \Rightarrow ({\mathbf{H}}_r(f) \: \text{is p.d.})$. In this case, the requirement that $\mathbf{L}$ has exactly $r$ distinct eigenvalues is already needed. 
\end{proof}

\section{Interpolation with graph basis functions} \label{sec:GBFinterpolation}

By the attribution $K_f(\node{v_i},\node{v_j}) = \mathbf{C}_{e_j} f (\node{v}_i)$ a p.d. function $f$ induces a p.d. kernel $K_f$. Therefore, we obtain an interpolation scheme for the generalized translates $\mathbf{C}_{e_j}f$ of the graph basis function $f$ by considering the corresponding kernel-based scheme introduced in Section \ref{sec:kernelmethods}. We summarize this interpolation scheme in Algorithm \ref{algorithm1}: 

\vspace{2mm}

\begin{algorithm}[H] \label{algorithm1}

\caption{Interpolation with Graph Basis Functions (GBF's)}

\vspace{4mm}

\KwIn{Signal values $x(\node{w}_1), \ldots, x(\node{w}_N)$ at the 
sampling nodes $W \subset V$. \newline A positive definite graph basis function $f \in \mathcal{P}_{+}$   
}

\vspace{2mm}

\textbf{Calculate} the $N$ generalized translates 
$\mathbf{C}_{e_{j_1}} f = e_{j_1} \ast f, \ldots, \mathbf{C}_{e_{j_N}} f = e_{j_N} \ast f$ with the correspondence $\node{v}_{j_k} = \node{w}_k$ for the nodes in $W$. 

\vspace{2mm}

\textbf{Solve} the linear system of equations

\begin{equation*} \label{eq:computationcoefficientsGBF} 
\underbrace{\begin{pmatrix} \mathbf{C}_{e_{j_1}}f(\node{w}_1) & \mathbf{C}_{e_{j_2}}f(\node{w}_1) & \ldots & \mathbf{C}_{e_{j_N}}f(\node{w}_1) \\
\mathbf{C}_{e_{j_1}}f(\node{w}_2) & \mathbf{C}_{e_{j_2}}f(\node{w}_2) & \ldots & \mathbf{C}_{e_{j_N}}f(\node{w}_2) \\
\vdots & \vdots & \ddots & \vdots \\
\mathbf{C}_{e_{j_1}}f(\node{w}_N) & \mathbf{C}_{e_{j_2}}f(\node{w}_N) & \ldots & \mathbf{C}_{e_{j_N}}f(\node{w}_N)
\end{pmatrix}}_{\mathbf{K}_{f,W}} 
\begin{pmatrix} c_1 \\ c_2 \\ \vdots \\ c_N \end{pmatrix}
= \begin{pmatrix} x(\node{w}_1) \\ x(\node{w}_2) \\ \vdots \\ x(\node{w}_N) \end{pmatrix}.
\end{equation*}

\textbf{Calculate} the GBF-interpolant
\[ \itpx(\node{v}) = \sum_{k=1}^N c_k \mathbf{C}_{e_{j_k}} f (\node{v}).\]

\end{algorithm}

\vspace{2mm}

In the following, we call $\itpx$ the GBF interpolant of the signal $x$ on the nodes $W$. The particular structure of this interpolation scheme allows us to discuss error estimates and stability issues in Section \ref{sec:errorbounds} very similarly to RBF interpolation in $\Rr^d$ or SBF interpolation on the unit sphere. The GBF interpolation space is spanned by the generalized translates $\mathbf{C}_{e_{j_k}} f$ of the GBF $f$, i.e.,
\[\mathcal{N}_{K_f,W} = \left\{x \in \mathcal{L}(G) \ | \ x = \sum_{k=1}^N c_k \mathbf{C}_{e_{j_k}} f \right\}. \]
Bochner's characterization in Theorem \ref{thm:Bochner} provides us with the following characterization of the inner product in the native space $\mathcal{N}_{K_f} = 
\mathcal{N}_{K_f,V}$.

\begin{theorem} \label{thm:nativespacecharacterization}
If $f \in \mathcal{P}_+$ then the inner product and the norm of the native space
$\mathcal{N}_{K_f}$ are given as
\[ \langle x , y \rangle_{K_f} = 
\sum_{k=1}^n \frac{\hat{x}_k \, \hat{y}_k}{\hat{f}_k} = \hat{y}^\intercal \mathbf{M}_{1/\hat{f}} \, \hat{x} \quad \text{and} \quad \| x \|_{K_f} = \sqrt{\sum_{k=1}^n \frac{\hat{x}_k^2}{\hat{f}_k}}. \]
\end{theorem}

\begin{proof}
By the characterization in Theorem \ref{thm:Bochner}, the eigendecomposition of the p.d. matrix $\mathbf{K}_f$ is given as $\mathbf{K}_f = \mathbf{U}\mathbf{M}_{\hat{f}}\mathbf{U^\intercal}$. This implies that the inverse $\mathbf{K}_f^{-1}$ of $\mathbf{K}_f$ is given by $\mathbf{K}_f^{-1} = \mathbf{U} \mathbf{M}_{1/\hat{f}}\mathbf{U^\intercal}$. Therefore, the inner product of the native space $\mathcal{N}_{K_f}$ can be written as
\[ \langle x , y \rangle_{K_f} = y^{\intercal} \mathbf{K}_f^{-1} x =
y^{\intercal} \mathbf{U} \mathbf{M}_{1/\hat{f}}\mathbf{U^\intercal} x = 
\hat{y}^\intercal \mathbf{M}_{1/\hat{f}} \, \hat{x} = \sum_{k=1}^n \frac{\hat{x}_k \, \hat{y}_k}{\hat{f}_k}.\]
\end{proof}

\section{Examples of positive definite functions on graphs} \label{sec:examples}

\begin{figure}[htbp]
	\centering
	\includegraphics[width= 1\textwidth]{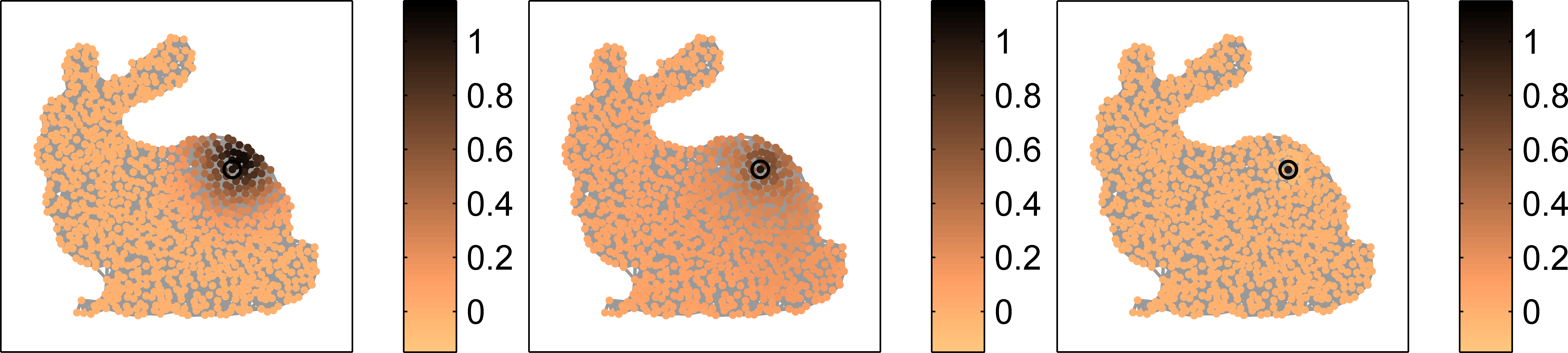} 
	\caption{Illustration of the shifts $\mathbf{C}_{e_j} f$ for different GBF's. Left: $f = f_{e^{- 10 \mathbf{L}}}$ in Example (5). Middle: $f = f_{\mathrm{pol},1}$ in Example (6). Right: $f=f_{\mathbf{L}^{(0)}}$ in Example (2). The ringed node corresponds to $\node{v}_j$.}
	\label{fig:examples}
\end{figure}

In the following, we list several important examples of p.d. GBF's on graphs. Some of them are related to well-known graph kernels. Generalized shifts $\mathbf{C}_{e_j} f$ of these GBF's are illustrated in Fig. \ref{fig:examples}.

\leavevmode
\begin{enumerate}
\item[(1)] (Trivial interpolation with the unity $\funit$) 
The unity $\funit = \sum_{k=1}^n u_k$ of the graph convolution is a p.d. function. We have $\mathbf{C}_{e_i} \funit = e_i$. Therefore, $\mathbf{K}_{\funit} = \mathbf{I}_n$ and $\mathbf{K}_{\funit,W} = \mathbf{I}_N$ are identity matrices and the interpolation space of the GBF $\funit$ is given by $\mathcal{N}_{K_{\funit},W} = \mathrm{span} \{e_{j_1}, \ldots, e_{j_N}\}$ (here, $\node{v}_{j_k}$ corresponds to the node $\node{w}_k$). Interpolation in terms of the basis functions $\{e_{j_1}, \ldots, e_{j_N}\}$ is nothing else than extending the given data $x(\node{w}_1), \ldots, x(\node{w}_N)$ by $x(\node{v}) = 0$ for all $\node{v} \in V \setminus W$.  
\item[(2)] (The graph Laplacian $\mathbf{L}$)
Our first more prominent example is the graph Laplacian $\mathbf{L}$ for a connected graph $G$. As the eigenvalues $\lambda_k$ of $\mathbf{L}$ are all positive except for $\lambda_1 = 0$, the graph Laplacian $\mathbf{L}$ corresponds to a c.p.d. kernel with the Mercer decomposition
\[\mathbf{L} = \sum_{k=2}^n \lambda_k u_k u_k^{\intercal}. \]
The Laplacian $\mathbf{L}$ is p.d. on the subspace $\mathrm{span} \{u_2, \ldots u_n\}$ and maps the constant signals $\mathrm{span} \{u_1\}$ to the zero signal $\zero$. Therefore, for every $\delta>0$ the augmented Laplacian
\[\mathbf{L}^{(\delta)} = \mathbf{L} + \delta u_1 u_1^{\intercal}\]
is positive definite. The p.d. generator $f_{\mathbf{L}^{(\delta)}}$ of the augmented
Laplacian is determined by
\[ \hat{f}_{\mathbf{L}^{(\delta)}} = (\delta,\lambda_2, \ldots, \lambda_n). \] 
\item[(3)] (Polynomials of the graph Laplacian $\mathbf{L}$)
The spectral calculus allows us to define further p.d. kernels based on the eigenvalue decomposition of $\mathbf{L}$. If $p_r$ is a positive polynomial of degree $r$ on the interval $[0,2]$, we get $p_r(\lambda_k)>0$ for all eigenvalues $\lambda_k$ of $\mathbf{L}$. Therefore, 
\[p_r(\mathbf{L}) = \sum_{k=1}^n p_r(\lambda_k) u_k u_k^{\intercal}\]
gives rise to a p.d. kernel on $G$. The Fourier transform of the generating GBF $f_{p_r(\mathbf{L})}$ is given as 
\[ \hat{f}_{p_r(\mathbf{L})} = (p_r(\lambda_1), \ldots, p_r(\lambda_n)). \]
These p.d. GBF's are relevant for practical applications, in particular if the size $n$ of $G$ gets large. In this case, the kernel matrix
$p_r(\mathbf{L})$ and its columns can be calculated quickly with simple 
matrix-vector multiplications based on the graph Laplacian $\mathbf{L}$.  
\item[(4)] (Variational or polyharmonic splines) Variational splines on graphs were introduced in \cite{Pesenson2009} as the solutions $\itpx$ of the interpolation problem \eqref{eq:interpolation} that minimize the functional $\|(\epsilon \mathbf{I}_n + \mathbf{L})^{s/2} \itpx\|$, $\epsilon > 0$, $s > 0$. 
In view of Proposition \ref{prop:minimumenergy}, this energy functional corresponds 
to the native space norm of the kernel 
$$ (\epsilon \mathbf{I}_n + \mathbf{L})^{-s} = \sum_{k=1}^n \frac{1}{(\epsilon + \lambda_k)^s} u_k u_k^{\intercal}.$$
Therefore, variational spline interpolation can be regarded as a GBF interpolation scheme based on the p.d. function $f_{(\epsilon \mathbf{I}_n + \mathbf{L})^{-s}}$ defined in the spectral domain as
\[ \hat{f}_{(\epsilon \mathbf{I}_n + \mathbf{L})^{-s}} = \ts \left(\frac{1}{(\epsilon + \lambda_1)^s}, \ldots, \frac{1}{(\epsilon + \lambda_n)^s}\right). \]
In \cite{Ward2018interpolating}, also the parameter choice $\epsilon = 0$ was considered. In this case, the functional 
$\| \mathbf{L}^{s/2} \itpx\|$ is a seminorm related to the p.s.d. kernel $(\mathbf{L}^\dagger)^{s}$. If the graph $G$ is connected the Fourier transform of the p.s.d. function $f_{(\mathbf{L}^{\dagger})^{s}}$ is given by
\[ \hat{f}_{(\mathbf{L}^\dagger)^{s}} = \ts \left(0, \lambda_2^{-s}, \ldots, \lambda_n^{-s}\right). \]
The GBF $f_{(\mathbf{L}^{\dagger})^{s}}$ is therefore a c.p.d. function with respect to the subspace $\mathrm{span} \{u_2, \ldots u_n\}$. In order to guarantee uniqueness, the interpolation problem can be treated as in Example (2) or as described in Section \ref{subsec:cpdkernels}. A more classical approach to solve the interpolation problem with variational splines is described in \cite{Ward2018interpolating} or in \cite{Hangelbroek2012} for related problems on manifolds. 

\item[(5)] (Diffusion kernels)
Diffusion kernels \cite{KondorLafferty2002} based on the Mercer decomposition
\[e^{ -t \mathbf{L}} = \sum_{k=1}^n e^{ -t \lambda_k} u_k u_k^{\intercal}\]
are as well p.d. for all $t \in \Rr$. The graph Fourier transform of the respective p.d. function $f_{e^{-t \mathbf{L}}}$ is given as
\[\hat{f}_{e^{-t \mathbf{L}}} = (e^{-t \lambda_1}, \ldots, e^{-t \lambda_n}). \]

\item[(6)] (Kernels with polynomial Fourier decay)
P.d. functions $f_{\mathrm{pol},s}$ with a polynomial decay on the spectrum $\hat{G}$ are determined by the Fourier coefficients
\[\hat{f}_{\mathrm{pol},s} = \left(1,\frac{1}{2^{s}},\frac{1}{3^s}, \ldots, 
\frac{1}{n^s} \right), \quad s > 0. \]
The corresponding kernel has the form 
\[\mathbf{K}_{f_{\mathrm{pol},s}} = \sum_{k=1}^n \frac{1}{k^s} u_k u_k^{\intercal}.\]
These p.d. functions are relevant for us in the discussion of error estimates for GBF interpolation. 

\item[(7)] (Bandlimited interpolation)
Interpolation in the space $\mathcal{B}_M$ of bandlimited functions can be described
with help of the p.s.d. function $f_{\mathbf{B}_M} = \sum_{k=1}^M u_k$, i.e., the unity element of the subalgebra $\mathcal{B}_M$. The corresponding kernel $\mathbf{B}_M = \sum_{k=1}^m u_k u_k^{\intercal}$ corresponds to the orthogonal projection onto the space $\mathcal{B}_M$. The fact that $f_{\mathbf{B}_M}$ is not strictly p.d. already indicates that we can not expect to have unisolvence for the interpolation problem\eqref{eq:interpolation} in the space $\mathcal{N}_{K_{f_{\mathbf{B}_M}}} = \mathcal{B}_M$. This is in fact also pointed out in \cite{Pesenson2008,TBL2016}.  
\end{enumerate}

\section{Space-frequency analysis with positive definite functions}
\label{sec:spacefrequency}

For a window function $f \in \mathcal{L}(G)$, the windowed Fourier transform 
$\mathbf{F}_{f} x $ of a signal $x$ is defined in the domain $G \times \hat{G}$ as (cf. \cite{shuman2012,shuman2016})
\begin{equation} \label{eq:windowedFouriertransform}
\mathbf{F}_{f} x (\node{v}_i,u_k) := \sqrt{n} \, x^{\intercal} (\mathbf{M}_{u_k} \mathbf{C}_{e_i} f).
\end{equation}
We reflected already on the role of the convolution $\mathbf{C}_{e_i} f$ as a generalized shift of the function $f$ on $G$. In a similar sense, the multiplication operator $\sqrt{n} \mathbf{M}_{u_k}$ in the windowed Fourier transform mimics a generalized modulation in terms of the Fourier basis $u_k$. The space-frequency analysis related to the windowed Fourier transform uses 
the coefficients $\mathbf{F}_{f} x (\node{v}_i,u_k)$ to decompose the signal $x$. In 
\cite{shuman2016} it is shown that the system 
$\{ \sqrt{n} \mathbf{M}_{u_k} \mathbf{C}_{e_i} f \ | \ i,k \in \{1, \ldots, n\}\}$ provides a frame for the space of signals $\mathcal{L}(G)$ if $\hat{f}_1 \neq 0$. 
If we assume that the window function $f$ is positive definite, we can tighten this statement. In this case the shifts $\{\mathbf{C}_{e_i} f \ | \ i \in \{1, \ldots, n\}\}$, form already a basis of $\mathcal{L}(G)$. In addition, we get the following result.

\begin{theorem}
Let $f \in \mathcal{P}_+$ and assume that $u_1$ is fixed as 
$u_1 = \frac{1}{\sqrt{n}}(1, \ldots, 1)^{\intercal}$. 
\begin{enumerate}
\item[(1)] If $u_k$ is non-vanishing for a fixed $k \in \{1, \ldots, n\}$, then the system
$\{ \sqrt{n} \, \mathbf{M}_{u_k} \mathbf{C}_{e_i} f \ | \ i \in \{1, \ldots, n\}\}$ is a basis of $\mathcal{L}(G)$. 
\item[(2)] Let $\{u_{k_1}, \ldots, u_{k_M}\}$ be a subset of $\hat{G}$ containing $M \leq n$ Fourier basis functions and $u_{k_1} = u_{1}$. Then 
$\{ \sqrt{n} \mathbf{M}_{u_{k_j}} \mathbf{C}_{e_i} \ | \ i \in \{1, \ldots, n\}, \ j \in \{1, \ldots, M\}\}$ is a frame for $G$ with the frame bounds
\[ \left( \min_{1 \leq k \leq n} \hat{f}_k \right)^2 \|x\|^2 \leq \sum_{i = 1}^n \sum_{j = 1}^M (x^{\intercal} (\sqrt{n} \, \mathbf{M}_{u_{k_j}} \mathbf{C}_{e_i} f))^2 \leq \sqrt{n} \left(\max_{1 \leq k \leq n} \hat{f}_k \right)^2 \|x\|^2.\]
\end{enumerate}
\end{theorem}  

\begin{proof} 
\noindent (1) As $f \in \mathcal{P}_+$ and $u_k(\node{v}) \neq 0$ for all $\node{v} \in V$, the matrices 
$\mathbf{K}_f$ and $\mathbf{M}_{u_k}$ are both invertible. Thus, $\sqrt{n} \mathbf{M}_{u_k} \mathbf{K}_f$ is invertible and the $n$ columns $\sqrt{n} \, \mathbf{M}_{u_k} \mathbf{C}_{e_i} f$, $i \in \{1, \ldots, n\}$, form a basis of $\mathcal{L}(G)$.

\noindent (2) By the characterization of p.d. functions in Theorem \ref{thm:Bochner} we know that $\mathbf{K}_f = \mathbf{U}\mathbf{M}_{\hat{f}}\mathbf{U^\intercal}$. Thus,
\[ \left( \min_{1 \leq k \leq n} \hat{f}_k\right) \|x\| \leq \| \mathbf{K}_f x\| \leq 
\left( \max_{1 \leq k \leq n} \hat{f}_k \right) \|x\|\]
holds true for all $x \in \Rr^n$. Now, if we use $u_{k_1} = u_1 = (1, \ldots, 1)^\intercal/ \sqrt{n}$, we get the lower bound
\begin{align*}\sum_{i = 1}^n \sum_{j = 1}^M (x^{\intercal} (\sqrt{n} \, \mathbf{M}_{u_{k_j}} \mathbf{C}_{e_i} f))^2 & \geq \sum_{i = 1}^n (x^{\intercal} (\sqrt{n} \, \mathbf{M}_{u_{1}} \mathbf{C}_{e_i} f))^2  = \|\mathbf{K}_f x \|^2 \geq 
\left( \min_{1 \leq k \leq n} \hat{f}_k\right)^2 \|x\|_2^2.
\end{align*}
On the other hand, we obtain as an upper bound
\begin{align*}\sum_{i = 1}^n \sum_{j = 1}^M (x^{\intercal} (\sqrt{n} \, \mathbf{M}_{u_{k_j}} \mathbf{C}_{e_i} f))^2 &= 
 \sqrt{n} \sum_{j = 1}^M  \sum_{i = 1}^n  ((\mathbf{M}_{u_{k_j}}x)^{\intercal} \mathbf{C}_{e_i} f))^2 = \sqrt{n} \sum_{j = 1}^M \|\mathbf{K}_f \mathbf{M}_{u_{k_j}}x\|^2
  \\ & \leq \left( \max_{1 \leq k \leq n} \hat{f}_k\right)^2 \sum_{j = 1}^M 
  \| \mathbf{M}_{u_{k_j}} x\|^2 \leq \left( \max_{1 \leq k \leq n} \hat{f}_k\right)^2 \sum_{j = 1}^M 
  ( u_{k_j}^{\intercal} x )^2 \leq \left( \max_{1 \leq k \leq n} \hat{f}_k\right)^2 \|x\|^2.
\end{align*}
\end{proof}  

\begin{remark} 
If $M = n$, a slight adaption of the last proof leads to the improved lower frame bound $\sqrt{n} \left( \min_{1 \leq k \leq n} \hat{f}_k\right)^2$. The upper and lower frame bounds $\sqrt{n} \left( \max_{1 \leq k \leq n} \hat{f}_k\right)^2$ and $\sqrt{n} \left( \min_{1 \leq k \leq n} \hat{f}_k\right)^2$ are optimal in this case. This was pointed out in \cite{shuman2016} with an equivalent formulation of the bounds.
\end{remark}

\section{Stability and Error Estimates} \label{sec:errorbounds}

The goal of this section is to find upper bounds for the
interpolation error $|x(\node{v}) - \itpx (\node{v})|$ and for the numerical condition of the interpolation. As before, $\itpx$ denotes the uniquely determined interpolant of a signal $x \in \mathcal{L}(G)$ in the native space $\mathcal{N}_{W,K_f}$. The approximation space $\mathcal{N}_{W,K_f}$ is defined upon the set $W = \{\node{w}_1, \ldots \node{w}_N\} \subset V$ of interpolation nodes and a positive definite function $f \in \mathcal{P}_+$ providing the kernel $K_f$ and the p.d. interpolation matrix $\mathbf{K}_{f,W}$.

\subsection{Norming Sets}

In order to measure the approximation quality of the interpolant $\itpx$, we need to know how well functions in subspaces of $\mathcal{L}(G)$ can be recovered from function samples on the node set $W$. This is provided by the following notion of norming set. As auxiliary subspaces of $\mathcal{L}(G)$, we consider the bandlimited signals $\mathcal{B}_M$ on $G$. 
We introduce also two projection operators $\mathbf{S}_W$ and $\mathbf{B}_M$ on $\mathcal{L}(G)$ as
\[\mathbf{S}_W x (\node{v}) = \left\{ \begin{array}{ll} 
x(\node{v}) & \text{if $\node{v} \in W$}, \\ 
0           & \text{otherwise},
           \end{array} \right. \quad \text{and} \quad \mathbf{B}_M x = \sum_{k=1}^M (u_k^\intercal x) \, u_k. \]
The mapping $\mathbf{S}_W$ projects a signal $x$ onto the subset $W \subset V$, whereas 
$\mathbf{B}_M$ describes the projection onto the space $\mathcal{B}_M$ of bandlimited functions.   

\begin{definition} \label{def:normingset}
We call a subset $W = \{\node{w}_1, \ldots, \node{w}_N\}$ of $V$ a norming set for the subspace $\mathcal{B}_M \subset \mathcal{L}(G)$ if the operator
$\mathbf{S}_W \mathbf{B}_M$ is injective on the subspace $\mathcal{B}_M$. 
\end{definition}
The injectivity of $\mathbf{S}_W \mathbf{B}_M$ on $\mathcal{B}_M$ guarantees that we have a well-defined inverse $(\mathbf{S}_W \mathbf{B}_M)^{-1}$ on the image $\mathbf{S}_W(\mathcal{B}_M) = \mathbf{S}_W \mathbf{B}_M (\mathcal{L}(G))$. The operator norm of this inverse
\[ \normconst = \sup_{z \in \mathbf{S}_W(\mathcal{B}_M), \|z\| \leq 1} \|(\mathbf{S}_W\mathbf{B}_M)^{-1}z\|\]
is referred to as norming constant of the set $W$. Definition \ref{def:normingset} is a formal way to describe the following equivalent statement: $W$ is a norming set for $\mathcal{B}_M$ if every bandlimited signal $x \in \mathcal{B}_M$ can be recovered 
uniquely from the samples $\{ x(\node{w}_1), \ldots, x(\node{w}_N) \}$. In the literature, the notion uniqueness set is therefore sometimes used instead of norming set, see \cite{Narang2013,Pesenson2008}.
  
There is a simple criterion derived in \cite{TBL2016} to see whether $W$ is a norming set for $\mathcal{B}_M$ or not. 

\begin{theorem} \label{res:critnormingset}
The set $W$ is a norming set for the space $\mathcal{B}_M$ of bandlimited functions if and only if the spectral norm of the matrix
$\mathbf{B}_M (\mathbf{I}_n - \mathbf{S}_W)\mathbf{B}_M$ is strictly less than $1$. 
The norming constant of the set $W$ is bounded by
\[ \normconst \leq \frac{1}{1-
\|\mathbf{B}_M (\mathbf{I}_n - \mathbf{S}_W)\mathbf{B}_M\|}. \]
\end{theorem}

\begin{proof} The first part of this statement is already proven in \cite[Theorem 4.1]{TBL2016}, the second part regarding the bound for the norming constant is new. In the following, we give (for the convenience of the reader) a combined proof for both parts. On several occasions, we will use the fact that $\mathbf{S}_W$ and $\mathbf{B}_M$ are both projection operators. 

The operator $\mathbf{S}_W \mathbf{B}_M$ restricted to the space $\mathcal{B}_M$ is injective if and only if 
$(\mathbf{S}_W \mathbf{B}_M)^\intercal \mathbf{S}_W \mathbf{B}_M = 
\mathbf{B}_M \mathbf{S}_W \mathbf{B}_M$ is invertible on $\mathcal{B}_M$ (we denote its pseudo-inverse by $(\mathbf{B}_M \mathbf{S}_W \mathbf{B}_M)^\dagger$). This is true if and only if the extended operator $\mathbf{I}_n - \mathbf{B}_M + \mathbf{B}_M \mathbf{S}_W \mathbf{B}_M$ is invertible on the entire space $\mathcal{L}(G)$.
This, on the other hand is equivalent to the fact that the spectral norm of the operator $\mathbf{B}_M - \mathbf{B}_M \mathbf{S}_W \mathbf{B}_M = \mathbf{B}_M (\mathbf{I}_n - \mathbf{S}_W)\mathbf{B}_M$ is strictly less than one. This shows the first statement. 

For the second statement, we write the inverse $(\mathbf{S}_W \mathbf{B}_M)^{-1}$ on 
the image $\mathbf{S}_W (\mathcal{B}_M)$ as
\[ (\mathbf{S}_W \mathbf{B}_M)^{-1}z = (\mathbf{B}_M \mathbf{S}_W \mathbf{B}_M)^{\dagger} \mathbf{B}_M \mathbf{S}_W z = (\mathbf{I}_n - \mathbf{B}_M + \mathbf{B}_M \mathbf{S}_W \mathbf{B}_M)^{-1} \mathbf{B}_M \mathbf{S}_W z. \]
We can now use the Neumann series expansion for $(\mathbf{I}_n - \mathbf{B}_M + \mathbf{B}_M \mathbf{S}_W \mathbf{B}_M)^{-1}$ and obtain
\[(\mathbf{S}_W \mathbf{B}_M)^{-1} z = \sum_{i=0}^{\infty} (\mathbf{B}_M - \mathbf{B}_M \mathbf{S}_W \mathbf{B}_M)^{i} \mathbf{B}_M \mathbf{S}_W z. \]
Then, taking the norm on both sides yields the desired bound
\[ \| (\mathbf{S}_W \mathbf{B}_M)^{-1} z \| \leq \sum_{i=0}^{\infty} \|\mathbf{B}_M - \mathbf{B}_M \mathbf{S}_W \mathbf{B}_M\|^{i} \| \mathbf{B}_M \mathbf{S}_W z \| \leq 
\frac{1}{1-
\|\mathbf{B}_M (\mathbf{I}_n - \mathbf{S}_W)\mathbf{B}_M\|} \|z\|. \]
\end{proof}

Theorem \ref{res:critnormingset} illustrates that the knowledge of whether $W$ is a norming set for $\mathcal{B}_M$ or not depends profoundly on the spectral structure $\hat{G}$ of the graph $G$ and is related to the existence of non-admissible regions in the combined space-frequency domain of the graph. In spectral graph theory, such regions describe uncertainty principles. For this link to uncertainty principles and concrete examples, consider \cite{TBL2016} and the more general framework in \cite{erb2019}.

\subsection{Main Error Estimate}

Our main error estimate reads as follows:

\begin{theorem} \label{thm-errorestimate}
Let $f \in \mathcal{P}_+$ and $W \subset V$ be a norming set for 
the space $\mathcal{B}_M$ on the graph $G$. Then, for the GBF interpolant $\itpx \in \mathcal{N}_{W,K_f}$ of a graph signal $x$ we get the uniform error bound
\[ \max_{\node{v} \in V}|x(\node{v}) - \itpx (\node{v})| \leq (1+\normconst) 
\left( \sum_{k=M+1}^n \hat{f}_k \right)^{1/2}\|x\|_{K_f}.\]
\end{theorem}

\noindent This error estimate is determined by three correlated factors: the norming constant $\normconst$, the
tail $\sum_{k=M+1}^n \hat{f}_k$ and the native space norm $\|x\|_{K_f}$. 
These three factors depend on the sampling set $W$, the bandwidth $M$, the decay of the Fourier coefficients $\hat{f}_k$, and the signal $x$. Regarding $M$ and the decay of the coefficients $\hat{f}_k$, we have a trade-off between the first two factors and the last two factors of the error estimate.
In general, we will obtain meaningful error estimates in Theorem \ref{thm-errorestimate} if the coefficients $\hat{f}_k$ decay rapidly, the signal $x$ is smooth (with small native space norm $\|x\|_{K_f}$) and the sampling set $W$ is a norming set for a large space $\mathcal{B}_M$.
If the Fourier coefficients $\hat{f}_k$ have a particular decay, we further obtain the following refinements:

\begin{corollary}
With the same assumptions as in Theorem \ref{thm-errorestimate}, we get the following bounds:
\begin{enumerate}
\item[(1)] If $\hat{f}_k \leq C_1 k^{-s}$, $s > 1$, then $ \ds \max_{\node{v} \in V}|x(\node{v}) - \itpx (\node{v})| \leq \sqrt{\ts \frac{C_1}{s-1}}\,(1+\normconst)\,  
M^{- \frac{s-1}{2}} \, \|x\|_{K_f}$.
\item[(2)] If $\hat{f}_k \leq C_2 e^{- t k}$, $t > 0$, then
$ \ds \max_{\node{v} \in V}|x(\node{v}) - \itpx (\node{v})| \leq \sqrt{\ts \frac{C_2}{1-e^{-t}}} \, (1+\normconst) \, e^{- \frac{t}{2} (M+1)} \, \|x\|_{K_f}$.
\end{enumerate}
\end{corollary}

\begin{proof}
This statement is an immediate consequence of Theorem \ref{thm-errorestimate} with the following observations:
\begin{enumerate}
\item[(1)] $ \ds \sum_{k=M+1}^n \hat{f}_k  \leq 
C_1 \!\!\! \sum_{k=M+1}^n \frac{1}{k^{s}} \leq C_1 \int_M^\infty \frac{1}{x^{s}}
\mathrm{d} x \leq C_1  \frac{M^{1-s}}{s-1} $.
\item[(2)] $ \ds \sum_{k=M+1}^n \hat{f}_k  \leq 
C_2 \!\!\! \sum_{k=M+1}^{\infty} e^{- t k } = C_2 \frac{e^{-t(M+1)}}{1 - e^{-t}} $.
\end{enumerate}

\end{proof} 

\subsection{Proof of Theorem \ref{thm-errorestimate}}
The proceeding in this proof is inspired by the proofs of similar error estimates for SBF's \cite{JeStWa99}, for positive definite kernels on Riemannian manifolds \cite{Dyn1999} and on compact groups \cite{erbfilbir2008}.

In order to estimate the error $|x(\node{v}) - \itpx (\node{v})|$, 
the first step of the proof is to represent the interpolant $\itpx \in \mathcal{N}_{K_f,W}$ in a suitable way. This representation is given in terms of a Lagrange-type basis
$\{ \ell_1, \ldots, \ell_N \}$ of $\mathcal{N}_{K_f,W}$ as
\begin{equation}\label{eq-Lagrange1}
\itpx(\node{v}) = \sum_{k=1}^N \ell_k(\node{v}) x(\node{w}_k).
\end{equation}
The Lagrange basis functions $\ell_k$ are determined as the interpolants 
\begin{equation}\label{eq-Lagrange2}
\ell_k(v) = \mathrm{I}_W e_{j_k} (\node{v}), \quad k \in \{1, \ldots, N\},
\end{equation}
where the node $\node{v}_{j_k}$ corresponds to the node $\node{w}_k \in W$. 
Now, using the fact that $K_f$ is the reproducing kernel of the Hilbert space $\mathcal{N}_{K_f}$, we obtain the estimate
\begin{align}\label{eq-error}
|x(\node{v}) - \itpx(\node{v})|
&=\ds |x(\node{v}) - \sum_{k=1}^N \ell_k(\node{v}) x(\node{w}_k)| = 
\left| \left\langle
  x, K_f(\cdot,\node{v}) - \sum_{k=1}^N \ell_k(\node{v}) K_f(\cdot,\node{w}_k) \right\rangle_{K_f} \right| \\
& \leq\ds \|x\|_{K_f} \left\| K_f(\cdot,\node{v}) - \sum_{k=1}^N \ell_k(\node{v}) K_f(\cdot,\node{w}_k) \right\|_{K_f} = \|x\|_{K_f} \left\| K_f(\cdot,\node{v}) - \sum_{k=1}^N \ell_k(\node{v}) \mathbf{C}_{e_{j_k}} f \right\|_{K_f}. \notag
\end{align}
The norm $P_{W,K_f}(\node{v})=\| K_f(\cdot,\node{v}) - \sum_{k=1}^N \ell_k(\node{v}) \mathbf{C}_{e_{j_k}} f \|_{K_f}$ is referred to
as power function in the RBF community, see \cite{Scha98,We05}.
It depends on the node $\node{v}$, the sampling nodes
$W$ and on the p.d. function $f$, but does not depend on the signal $x$. 
We will conclude this proof by estimating the power function $P_{W,K_f}(\node{v})$.
For this, we need two well-known auxiliary results. The first is related to the power function.

\begin{lemma} (\cite[Theorem 11.1]{Scha98}, \cite[Theorem 11.5]{We05})\label{lem-powerfunction}
If $f \in \mathcal{P}_+$, then $\sum_{k=1}^N \ell_k(\node{v}) \mathbf{C}_{e_{j_k}} f$ is the best approximation of $K_f(\cdot,\node{v})$ in the subspace $\mathcal{N}_{K_f,W}$ with respect to the native space norm in $\mathcal{N}_{K_f}$. 
\end{lemma}

\begin{proof} This result is a consequence of the orthogonality of the subspace $\mathcal{N}_{K_f,W}$ to the vector $K_f(\cdot,\node{v}) - \sum_{k=1}^N \ell_k(\node{v}) \mathbf{C}_{e_{j_k}} f$. This follows from the identities
\begin{align*}  \left\langle K_f(\cdot,\node{v}) - \sum_{k=1}^N \ell_k(\node{v}) \mathbf{C}_{e_{j_k}} f, \mathbf{C}_{e_{j_i}} f \right\rangle_{K_f}
 &= K_f(\node{w}_i,\node{v}) - \sum_{k=1}^N \ell_k(\node{v}) \mathbf{C}_{e_{j_k}} f(\node{w}_i) \\
 &= \mathbf{C}_{e_{j_i}} f(\node{v}) - \sum_{k=1}^N \ell_k(\node{v}) \mathbf{C}_{e_{j_k}} f(\node{w}_i) = 0.
\end{align*}
The last equality follows from that fact that, by the definition of the Lagrange basis $\ell_k$, $k \in \{1, \ldots N\}$, the
function $\sum_{k=1}^N \ell_k(\node{v}) \mathbf{C}_{e_{j_k}} f(\node{w}_i)$ 
interpolates $\mathbf{C}_{e_{j_i}} f$ at all nodes $\node{w} \in W$. As this interpolant is unique in $\mathcal{N}_{K_f,W}$, the sum $\sum_{k=1}^N \ell_k(\node{v}) \mathbf{C}_{e_{j_k}} f(\node{w}_i)$ corresponds to $\mathbf{C}_{e_{j_i}} f(\node{v})$ on the entire node set $V$. 
\end{proof}

The second auxiliary result is related to norming sets. It can be proven with a functional analytic argument including the Hahn-Banach theorem. The details are given in \cite[Theorem 3.4]{We05}.

\begin{lemma} (\cite[Theorem 3.4]{We05}) \label{lem-normingset} Suppose $W = \{w_1, \ldots, w_N\}$ is a norming set for $\mathcal{B}_M \subset \mathcal{L}(G)$. Then, for every node $\node{v} \in V$, there are coefficients $(a_1(\node{v}), \ldots, a_N(\node{v}))\in \Rr^N$ such that
\[  x(\node{v}) = \sum_{k=1}^N a_k(\node{v}) x(\node{w}_k) \quad \text{and} \quad 
\sum_{k=1}^N |a_k(\node{v})|^2 \leq \normconst^2
\]
for all signals $x \in \mathcal{B}_M$.
\end{lemma}

{\noindent \bfseries Proof of Theorem \ref{thm-errorestimate}.}
Starting from the bound of the approximation error given in \eqref{eq-error}, we continue to estimate the power function $P_{W,K_f}(\node{v})$. 
Without loss of generality we assume that $\node{v} \notin W$ (for $\node{v} \in W$ the power function is zero) and set $\node{w}_0=\node{v}_{j_0} = \node{v}$ as well as $\ell_0(\node{v})=-1$. 
Then, by using the characterization of the native space norm given in Theorem \ref{thm:nativespacecharacterization}, we can rewrite the square of the power function as
\begin{align*}
P_{W,K_f}^2(\node{v})
&= \left\| K_f(\cdot,\node{v}) - \sum_{k=1}^N \ell_k(\node{v}) K_f(\cdot,\node{w}_k) \right\|_{K_f}^2 =
\left\| \sum_{k=0}^N \ell_{k}(\node{v}) \mathbf{C}_{e_{j_k}} f \right\|_{K_f}^2  = \sum_{l=1}^n
\hat{f}_l  \left(\sum_{k=0}^N \ell_{k}(\node{v})(\widehat{e_{j_k}})_l\right)^2.
\end{align*}
By Lemma \ref{lem-powerfunction}, the square $P_{W,K_f}^2(\node{v})$ is minimized as a functional by the coefficients $\ell_{k}(\node{v})$. 
Therefore, we obtain an upper bound of $P_{W,K_f}(\node{v})^2$ by replacing the coefficients $\ell_{k}(\node{v})$ with the
functions $a_k(\node{v})$, $k\in \{1, \ldots, N\}$, given in Lemma \ref{lem-normingset}. In addition, we set $a_0(\node{v})= -1$. In this way, we get the bound
\begin{align*}
P_{X,K_f}^2(x) &\leq \sum_{l=M+1}^n \hat{f}_l  \left(\sum_{k=0}^N a_{k}(\node{v})(\widehat{e_{j_k}})_l\right)^2
\leq \sum_{l=M+1}^n \hat{f}_l  \sum_{k=0}^N a_{k}^2(\node{v}) \sum_{k=0}^N(\widehat{e_{j_k}})_l^2
\\ &\leq \sum_{l=M+1}^n \hat{f}_l  \left( 1 + \sum_{k=1}^N |a_{k}(\node{v})|^2\right) 
\leq (1 +  \normconst)^2 \sum_{l=M+1}^n \hat{f}_l. 
\end{align*}
Taking the square root on both sides, we obtain precisely the statement of the theorem. \qed

\subsection{Stability}

A common measure for the absolute numerical condition of a linear interpolation scheme $\itpx$ is given by the operator norm $\sup_{\|x\| \leq 1 } \|\itpx\|$. It describes the worst case amplification of errors in the sampling data by the interpolation process. For this numerical condition number, we get:

\begin{theorem} \label{thm:stability}
If $f \in \mathcal{P}_+$, then the numerical condition number for GBF interpolation is bounded by
\[ \sup_{\|x\| \leq 1 } \|\itpx\| \leq \|\mathbf{K}_{f,W}^{-1}\| \|\mathbf{K}_{f}\| \leq \frac{\max_{1 \leq k \leq n} \hat{f}_k}{\min_{1 \leq k \leq n} \hat{f}_k}.\]
\end{theorem}

\begin{proof}
In matrix-vector notation we can write the interpolant $\itpx$ compactly as $$\itpx = 
(\mathbf{C}_{e_{j_1}}f, \ldots, \mathbf{C}_{e_{j_N}}f )\mathbf{K}_{f,W}^{-1} (x_{j_1}, \ldots, x_{j_N})^{\intercal}.$$ We can therefore bound the norm $\|\itpx\|$ by 
\begin{align*}
\|\itpx\| &\leq \|(\mathbf{C}_{e_{j_1}}f, \ldots, \mathbf{C}_{e_{j_N}}f )\| \|\mathbf{K}_{f,W}^{-1}\| \| (x_{j_1}, \ldots, x_{j_N})^{\intercal}\|
\leq \|\mathbf{K}_f\| \|\mathbf{K}_{f,W}^{-1}\| \|x\|.
\end{align*}
By Theorem \ref{thm:Bochner}, we know that the Fourier coefficients of $f$ are the eigenvalues of the p.d. matrix $\mathbf{K}_f$ and thus
$$ \| \mathbf{K}_f \| = \max_k \hat{f}_k \quad \text{and} \quad 
\| \mathbf{K}_f^{-1} \| = \frac{1}{\min_k \hat{f}_k}.$$ 
Further, by the inclusion principle \cite[Theorem 4.3.15]{HornJohnson1985}, the smallest eigenvalue of the principal submatrix $\mathbf{K}_{f,W}$ is larger than $\min_k \hat{f}_k$. We therefore get $\| \mathbf{K}_{f,W}^{-1} \| \leq (\min_k \hat{f}_k)^{-1}$, and, thus, the statement of the theorem. 
\end{proof}

Therefore, stability gets to an issue for GBF interpolation as soon as the interpolation matrix $\mathbf{K}_{f,W}$ is badly conditioned. Choosing basis functions $f$ in which the Fourier coefficients $\hat{f}_k$ are all distant from $0$ avoids bad conditioning. On the other hand, for the error bounds in Theorem \ref{thm-errorestimate} it is relevant that the Fourier coefficients decay rapidly towards $0$. This can be regarded as a trade-off between stability and approximation quality of the scheme and is a phenomenon that is typically encountered also in classical RBF and SBF interpolation as, for instance, discussed in \cite{demarchi2010,Fasshauer2011,We05}.     

\subsection{Numerical Example}

To get an impression on how GBF interpolation performs in comparison to a pure bandlimited interpolation, we give two numerical examples. The test graph $G$ is a reduced point cloud extracted from the Stanford bunny (Source: Stanford University Computer Graphics Laboratory). It contains $n = 900$ nodes projected in the $xy$-plane and $7325$ edges. Two nodes are therein connected with an edge, if the euclidean distance between the nodes is smaller than a given radius of $0.01$. We recursively construct a sequence $W_N$ of $N$ sampling sets in $V$ such that $\# W_N = N$, $W_{N-1}$ is contained in $W_N$, and the new node $\node{w}_N$ in $W_N$ is chosen randomly from $V \setminus W_{N-1}$. As a first test signal, we use the signal $x^{(1)} = u_4$ illustrated in Fig. \ref{fig:approximation1}, i.e. a bandlimited test function in the space $\mathcal{B}_4$. As a second example, we use a non-bandlimited, smooth signal $x^{(2)}$ shown 
in Fig. \ref{fig:approximation2}. The Fourier coefficients $\hat{x}^{(2)}_k$ of $x^{(2)}$ are decaying exponentially in $k$.  

The signal $x^{(1)} = u_4$ can be recovered exactly in the space $\mathcal{B}_N$ if $N \geq 4$ and $W_N$ is a norming set for $\mathcal{B}_N$. This is visible in Fig. \ref{fig:approximation1}. On the other hand, we see in Fig. \ref{fig:approximation2} that interpolation in $\mathcal{B}_N$ gets highly unstable if the signal $x^{(2)}$ is outside of $\mathcal{B}_N$ also if $x^{(2)}$ is very smooth. The GBF interpolants on the other hand show a similar stable behavior in both cases. Also, as predicted by Theorem \ref{thm-errorestimate}, the results in Fig. \ref{fig:approximation1} and \ref{fig:approximation2} show that the Fourier decay of the various GBF's has a strong impact on the convergence rates if the interpolated signals are smooth.

\begin{figure}[htb]
	\centering
	\hspace{-0.8cm}
	\begin{minipage}[t]{0.3\textwidth} \includegraphics[height = 4.8cm]{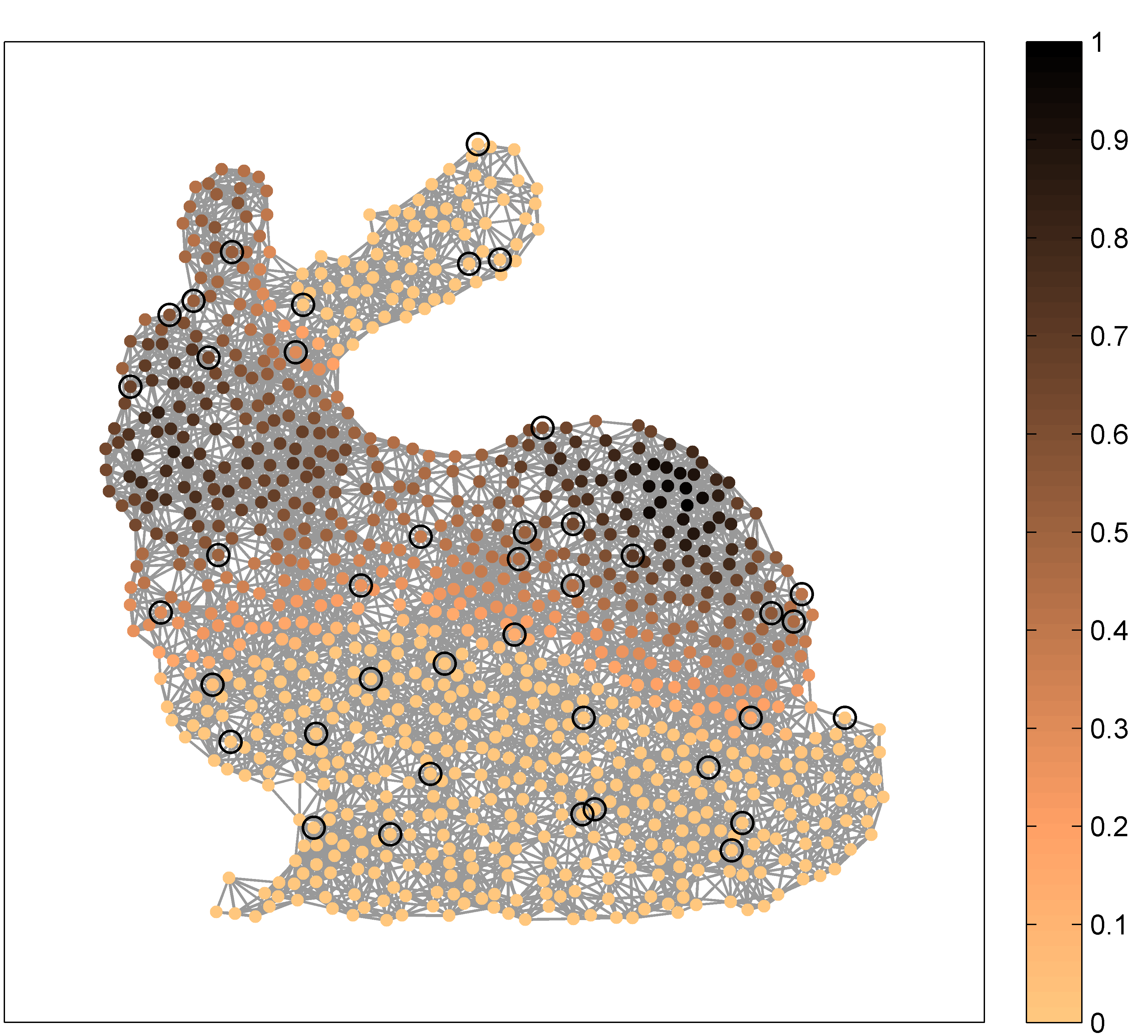}
	\end{minipage}
	\hspace{-0.01cm}	
	\begin{minipage}[t]{0.3\textwidth} \includegraphics[height = 4.8cm]{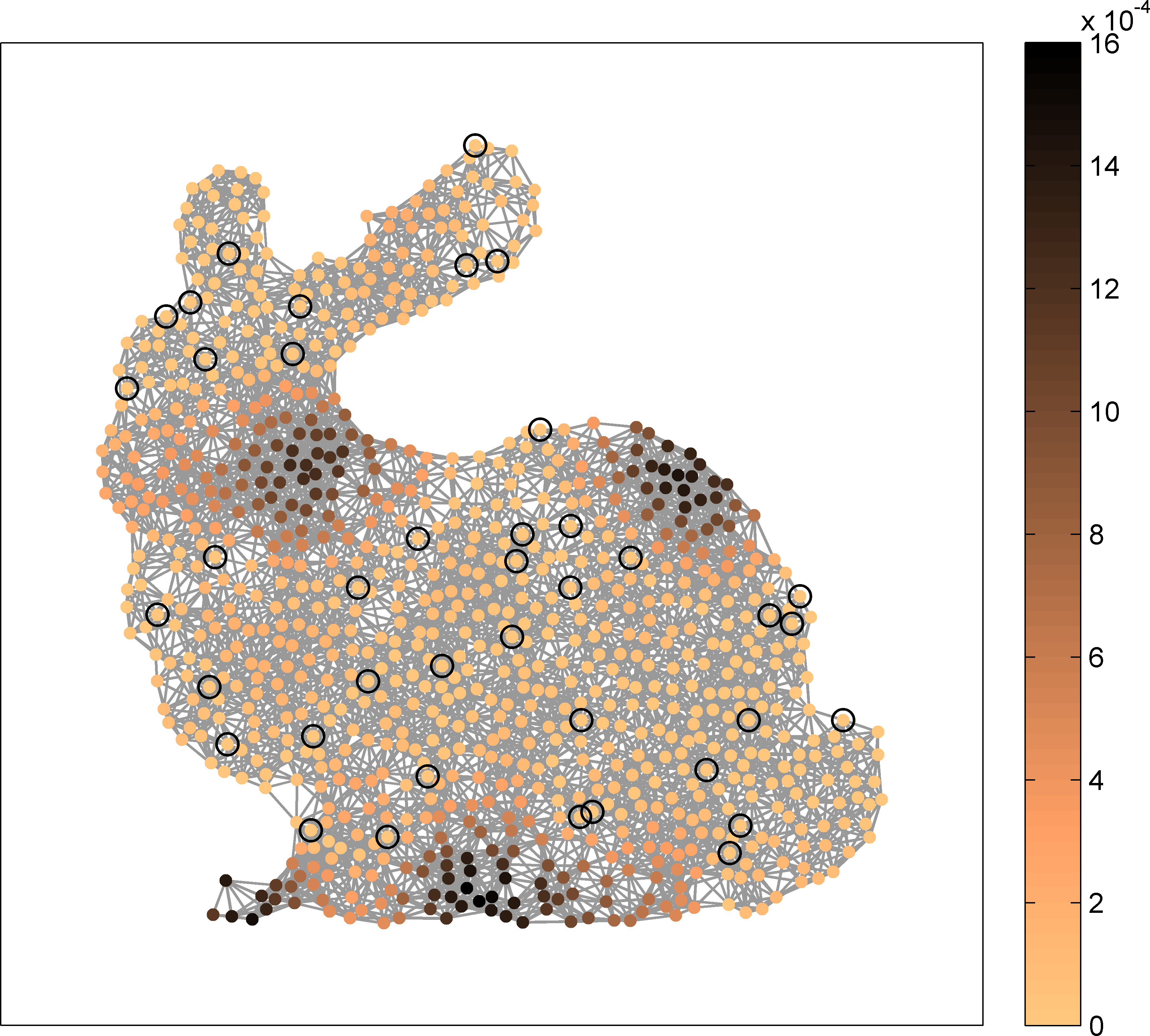}
	\end{minipage}
	\hspace{-0.1cm}		
	\begin{minipage}[t]{0.3\textwidth} \includegraphics[height = 4.8cm]{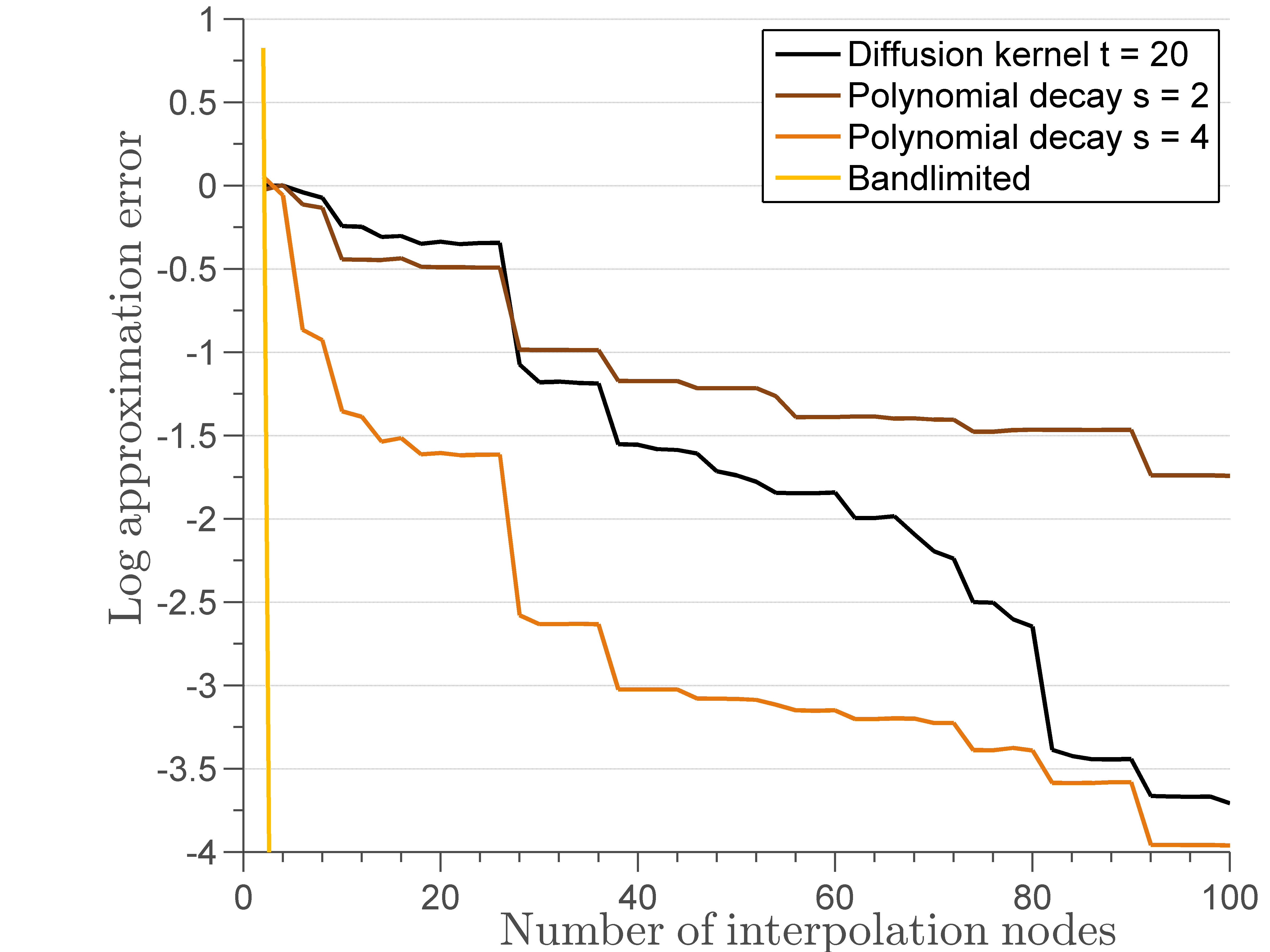}
	\end{minipage}
	\caption{GBF interpolation for the input signal $x^{(1)} = u_4$. Left: GBF interpolant for the nodes $W_{40}$ and the GBF $f_{\mathrm{pol},4}$ given in Example (6). Middle: interpolation error with respect to the original signal. Right: Interpolation errors for GBF schemes in terms of the number $N$ of interpolation nodes.}
  	 \label{fig:approximation1}
\end{figure}

\begin{figure}[htb]
	\centering
	\hspace{-0.8cm}
	\begin{minipage}[t]{0.3\textwidth} \includegraphics[height = 4.8cm]{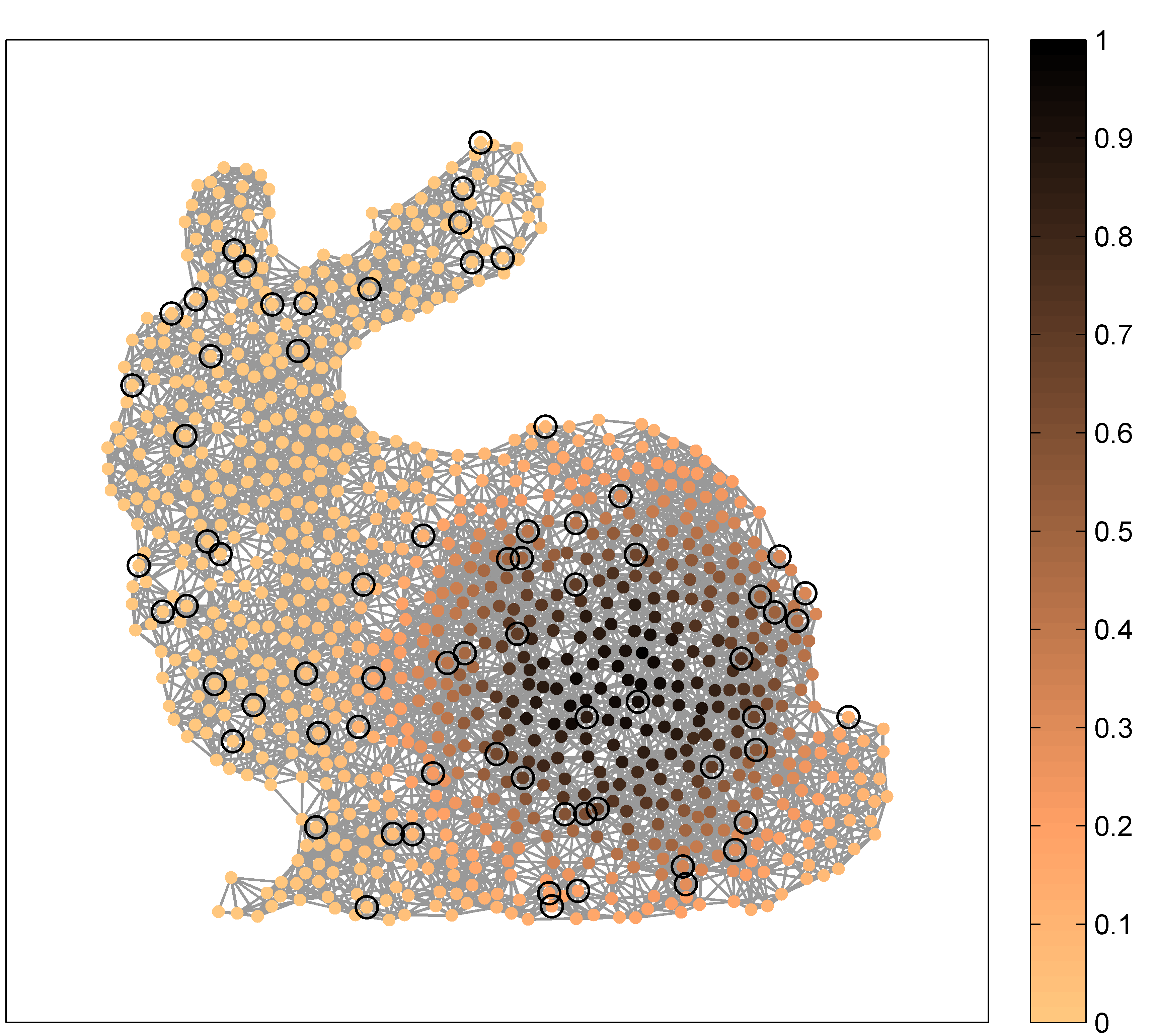}
	\end{minipage}
	\hspace{-0.01cm}	
	\begin{minipage}[t]{0.3\textwidth} \includegraphics[height = 4.8cm]{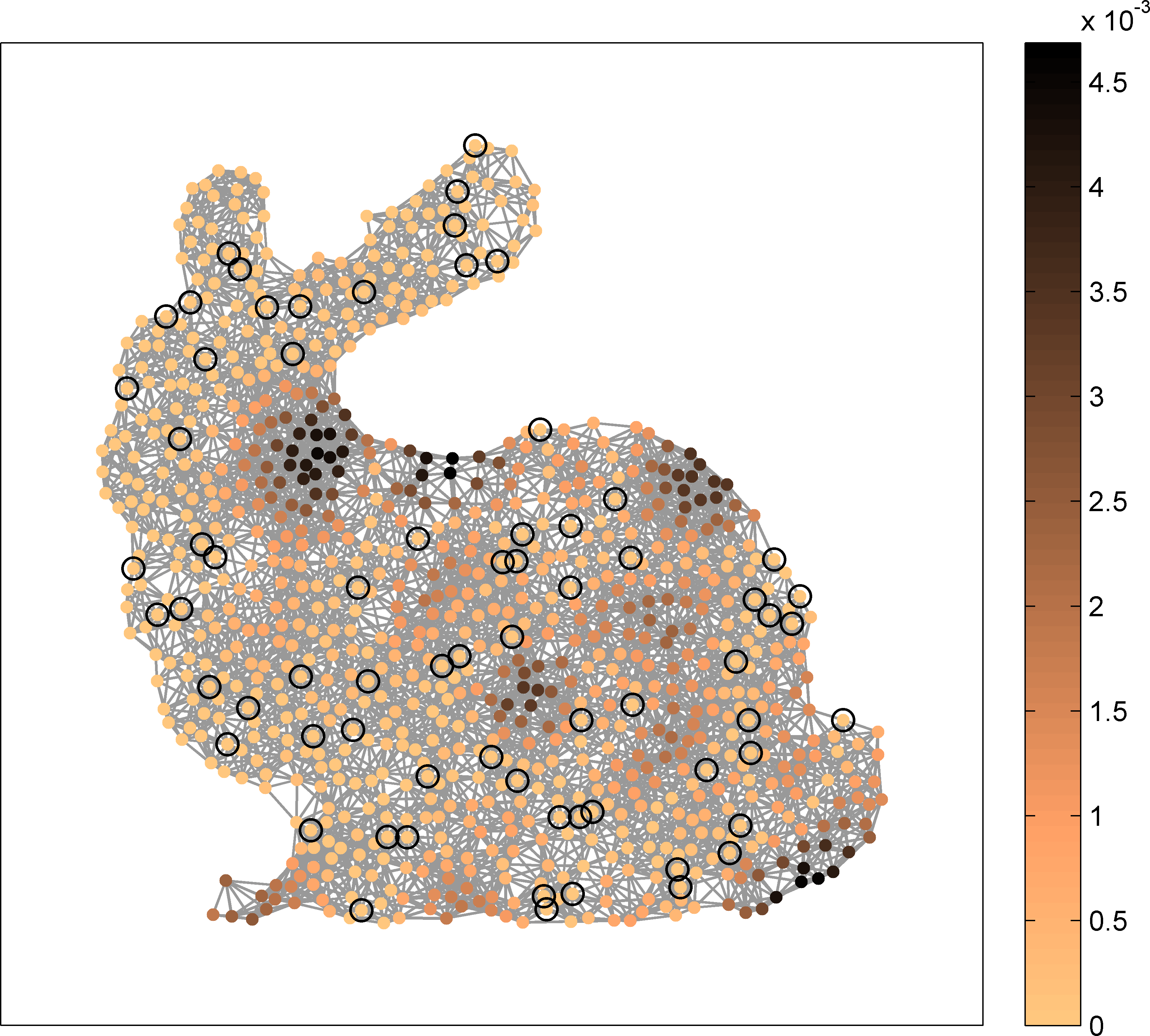}
	\end{minipage}
	\hspace{-0.1cm}		
	\begin{minipage}[t]{0.3\textwidth} \includegraphics[height = 4.8cm]{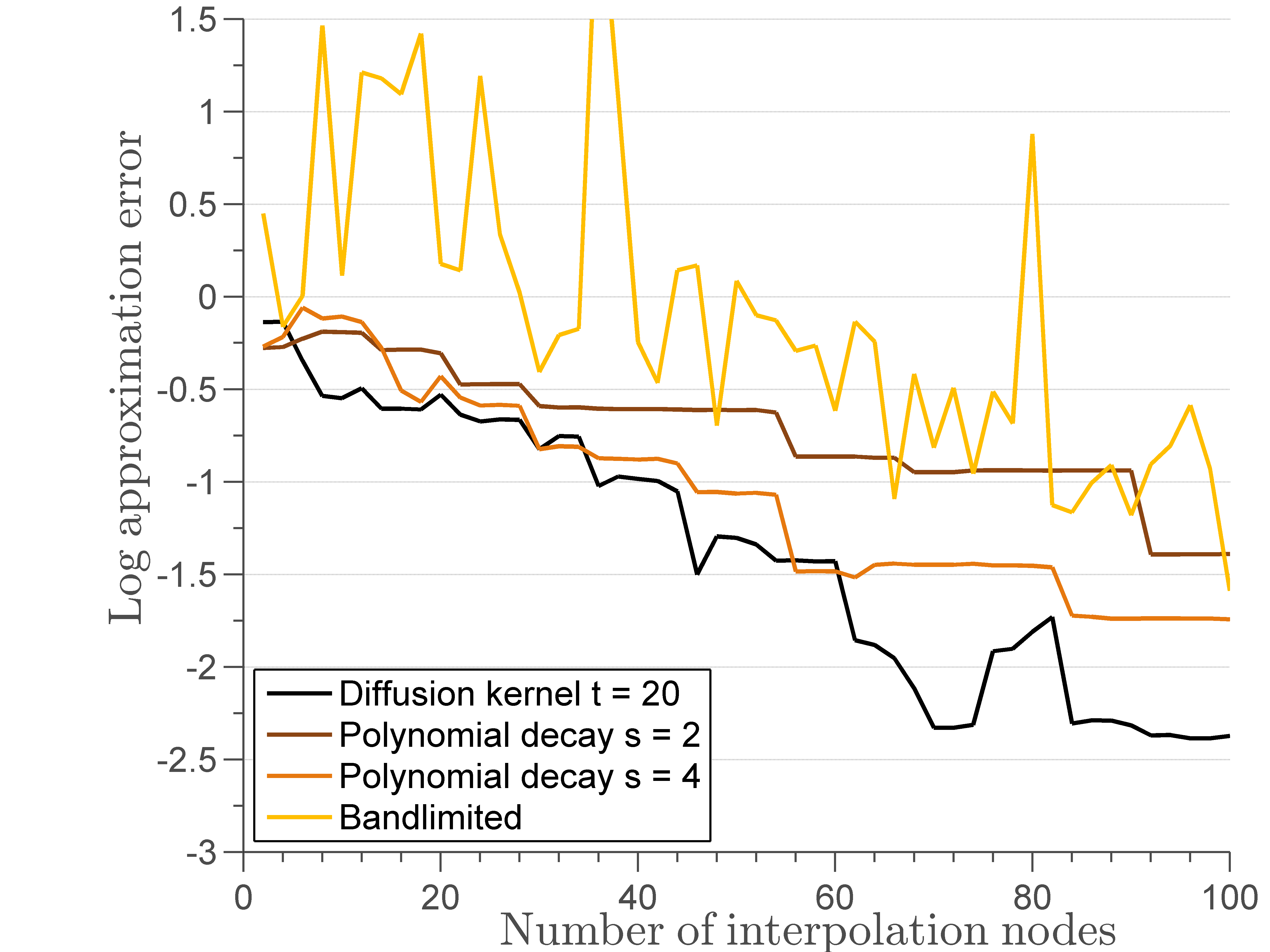}
	\end{minipage}
\caption{GBF interpolation for the input signal $x^{(2)}$. Left: GBF interpolant for the nodes $W_{70}$ and the diffusion GBF $f_{e^{- 20 \mathbf{L}}}$ of Example (5). Middle: interpolation error with respect to the original signal. Right: Interpolation errors for GBF schemes in terms of the number $N$ of interpolation nodes.}
  	 \label{fig:approximation2}
\end{figure}

\section{Integration of graph signals with positive definite functions}
\label{sec:integration}

As a final application of p.d. functions on graphs, we are interested in finding quadrature weights $\mu_k$, $k \in \{1, \ldots, N\}$ such that the integration functional $\frac{1}{n} \sum_{i=1}^n x(\node{v}_i)$ is well approximated by a sum of the form $\sum_{k=1}^N \mu_k x(\node{w}_k)$. Again $W = \{\node{w}_1, \ldots \node{w}_N\}$ is a subset of $V$. Similarly, as proposed for variational splines \cite{Pesenson2011}, we construct the quadrature weights in such a way that the quadrature formula is exact for all signals in the interpolation space $\mathcal{N}_{K_f,W}$, i.e.
\begin{equation} \label{eq-intexact} \frac{1}{n} \sum_{i=1}^n x(\node{v}_i) = \sum_{k=1}^N \mu_k x(\node{w}_k) \quad \text{for all} \; x \in \mathcal{N}_{K_f,W}.
\end{equation}
As before, $f$ is a p.d. GBF providing the basis
$\{\mathbf{C}_{e_{j_1}} f, \mathbf{C}_{e_{j_2}} f, \ldots, \mathbf{C}_{e_{j_N}} f\}$ for the space $\mathcal{N}_{K_f,W}$. The indices $j_k$ are determined by the relation
$\node{v}_{j_k} = \node{w}_k$. 

The exactness in \eqref{eq-intexact} provides us with a system of equations to determine the coefficients $\mu_k$, $k \in \{1, \ldots, N\}$. To derive this system we assume that the first eigenvector $u_1$ of the Laplacian is given by $u_1 = (1,\ldots,1)^{\intercal}/\sqrt{n}$. Then, plugging the basis functions $\mathbf{C}_{e_{j_k}} f$ into equation \eqref{eq-intexact}, we get the identities
\begin{align*} 
 \frac{1}{\sqrt{n}} u_{j_k}(\node{v}_1) \hat{f}_1 = \frac{1}{\sqrt{n}} (\widehat{\mathbf{C}_{e_{j_k}} f})_1= \frac{1}{\sqrt{n}} u_1^{\intercal} \mathbf{C}_{e_{j_k}} f  = \frac{1}{n}\sum_{i=1}^n \mathbf{C}_{e_{j_k}} f (\node{v}_i)
 = \sum_{l=1}^N \mu_l \mathbf{C}_{e_{j_k}} f(\node{w}_l)
\end{align*}
for $k \in \{1, \ldots, N\}$. Combining these $N$ identities and using the fact that
$\mathbf{C}_{e_{j_k}} f(\node{w}_l) = \mathbf{C}_{e_{j_l}} f(\node{w}_k)$, we get the linear system of equations 
\begin{equation} \label{eq:computationquadraturerule} 
\underbrace{\begin{pmatrix} \mathbf{C}_{e_{j_1}} f(\node{w}_1) & \cdots & \mathbf{C}_{e_{j_N}} f(\node{w}_1) \\
\vdots & \ddots & \vdots \\
\mathbf{C}_{e_{j_1}} f(\node{w}_N) & \cdots & \mathbf{C}_{e_{j_N}} f(\node{w}_N)
\end{pmatrix}}_{\mathbf{K}_{f,W}} \begin{pmatrix} \mu_1 \\ \vdots \\ \mu_N \end{pmatrix}
= \frac{1}{\sqrt{n}} \hat{f}_1 \begin{pmatrix} u_{j_1}(\node{v}_1) \\ \vdots \\ u_{j_N}(\node{v}_1) \end{pmatrix}.
\end{equation}
As $f$ is p.d., the matrix $\mathbf{K}_{f,W}$ is invertible and the coefficients $\mu_1, \ldots, \mu_N$ are uniquely determined. 

\begin{corollary} \label{thm-errorestimatequadrature}
Let $f \in \mathcal{P}_+$ and $W \subset V$ be a norming set for 
the space $\mathcal{B}_M$ on the graph $G$. Further, let the quadrature rule $\mathrm{Q}_W x = \sum_{k=1}^N \mu_k x(\node{w}_k)$ be exact for all signals in $\mathcal{N}_{K_f,W}$. Then, for $x \in \mathcal{L}(G)$, we have the error bound
\[ \left| \frac{1}{n} \sum_{i=1}^n x(\node{v}_i) - \mathrm{Q}_W x \right| \leq (1+\normconst) 
\left( \sum_{k=M+1}^n \hat{f}_k \right)^{1/2}\|x\|_{K_f}.\]
\end{corollary}

\begin{proof}
As the quadrature formula is exact for all elements of $\mathcal{N}_{K_f,W}$, we get for the interpolant $\itpx \in \mathcal{N}_{K_f,W}$ of a signal $x$ 
the identities $$\mathrm{Q}_W x = \mathrm{Q}_W \itpx = \frac{1}{n} \sum_{i=1}^n \itpx (\node{v}_i).$$
Therefore, 
\[ \left| \frac{1}{n} \sum_{i=1}^n x(\node{v}_i) - \mathrm{Q}_W x \right| 
= \left| \frac{1}{n} \sum_{i=1}^n ( x(\node{v}_i) - \itpx(\node{v}_i)) \right| \leq \max_{\node{v} \in V}|x(\node{v}) - \itpx (\node{v})|,\]
and the stated bound follows by Theorem \ref{thm-errorestimate}.
\end{proof}

For the variational spline kernel $f_{(\epsilon \mathbf{I}_n + \mathbf{L})^{-s}}$ considered in Example (4) the bound in Corollary \ref{thm-errorestimatequadrature} seems to be complementary to the quadrature error given in \cite[Theorem 3.3]{Pesenson2011}. While in \cite{Pesenson2008,Pesenson2009,Pesenson2011} a $\Lambda$-set terminology is used to describe the interpolation and quadrature quality of variational splines, we used the complementary notion of norming sets for the bounds in Theorem \ref{thm-errorestimate} and  Corollary \ref{thm-errorestimatequadrature}. For bandlimited functions, a further interesting quadrature rule related to kernels based on powers of the graph Laplacian is derived in \cite{Linderman2018}.

\section*{Acknowledgment}
This work was partially supported by GNCS-In$\delta$AM and by the European Union's Horizon 2020 research and innovation programme ERA-PLANET, grant agreement no. 689443.


\begin{thebibliography}{1}
\scriptsize

\bibitem{Aronszajn1950}
{\sc Aronszajn, N.} 
\newblock{ Theory of reproducing kernels.}
\newblock {\em Trans. Amer. Math. Soc. 68} (1950), 337–404.

\bibitem{Belkin2006}
{\sc Belkin, M., Niyogi, P. and Sindhwani, V.} 
\newblock {Manifold Regularization: A Geometric Framework for Learning from Labeled and Unlabeled Examples.}
\newblock {\em J. Mach. Learn. Res. 7} (2006), 2399--2434.

\bibitem{BelkinNiyogi2004}
{\sc Belkin, M. and Niyogi, P. } 
\newblock{ Semi-supervised learning on Riemannian manifolds.}
\newblock {\em Machine Learning 56}, 1-3 (2004), 209--239.

\bibitem{BerschneiderCastell2009}
{\sc Berschneider, G., and zu Castell, W.} 
\newblock{Conditionally positive definite kernels and Pontryagin spaces.}
\newblock in {\em M. Neamtu, L.L. Schumaker, (eds.), Approximation Theory XII} Nashboro Press, Brentwood, TN, (2008), 27--37.

\bibitem{Bochner1933}
{\sc Bochner, S.} 
\newblock{Monotone Funktionen, Stieltjes Integrale und harmonische Analyse.}
\newblock {\em Math. Ann. 108} (1933), 378–-410.

\bibitem{buhmann2003}
{\sc Buhmann, M.}
\newblock {\em Radial Basis Functions: Theory and Implementations}.
\newblock {Cambridge University Press}, 2003. 

\bibitem{Chenetal2015}
{\sc Chen, S., Varma, R., Sandryhaila, A., and Kovačević, J.}
\newblock {Discrete Signal Processing on Graphs: Sampling Theory.} 
\newblock {\em IEEE Transactions on Signal Processing 63}, 24 (2015), 6510--6523.
 
\bibitem{Chung}
{\sc Chung, F.R.K.}
\newblock {\em Spectral Graph Theory}.
\newblock {American Mathematical Society}, Providence, RI, 1997.

\bibitem{Davidson1996}
{\sc Davidson, K.R.}
\newblock {\em $C^{\ast}$-Algebras by Example}. Fields Institute Monographs, 
\newblock {American Mathematical Society}, Providence, RI, 1996. 

\bibitem{Dyn1999}
{\sc Dyn, N., Narcowich, F.J. and Ward, J.D.}
\newblock {Variational Principles and Sobolev-Type Estimates for Generalized Interpolation on a Riemannian Manifold}.
\newblock {\em J. Constr. Approx. 15} (1999), 175--208.

\bibitem{demarchi2010}
{\sc De Marchi, S. and Schaback, R.}
\newblock  Stability of kernel-based interpolation.
\newblock {\em Adv. Comput. Math. 32}, 2 (2010), 155–161.
  
\bibitem{erb2019}
{\sc Erb, W.}
\newblock {Shapes of Uncertainty in Spectral Graph Theory}.
\newblock {\em 	arXiv:1909.10865 \/} (2019).

\bibitem{erbfilbir2008}
{\sc Erb, W., and Filbir, F.}
\newblock {Approximation by positive definite functions on compact groups}.
\newblock {\em 	Numer. Funct. Anal. Optim. 29}, (9-10) (2019), 1082--1107.

\bibitem{Fasshauer2011}
{\sc Fasshauer, G.E.}
\newblock {Positive definite kernels: past, present and future}.
\newblock {\em 	Dolomites Res. Notes. Approx. 4}, (2011), 21-63.

\bibitem{Hangelbroek2012}
{\sc Hangelbroek, T., Narcowich, F.J., and Ward, J.D.}
\newblock{Polyharmonic and Related Kernels on Manifolds: Interpolation and Approximation.}
\newblock{Found. Comput. Math. 12}, 5 (2012), 625--670.

\bibitem{HornJohnson1985}
{\sc Horn, R.A., and Johnson, C.R.}
\newblock {\em Matrix Analysis},
\newblock Cambridge University Press, 1985.

\bibitem{Hubbert2005}
{\sc Hubbert, S., Le Gia, Q.T., Morton, T.M. }
\newblock {\em Spherical Radial Basis Functions, Theory and Applications},
\newblock Springer International Publishing, 2015.

\bibitem{JeStWa99} 
{\sc Jetter, K., Stöckler, J., and Ward, J.D.}
\newblock {Error estimates for scattered data interpolation on spheres.}
\newblock {\em Math. Comput. 68\/}, 226 (1999), 733--747.

\bibitem{KondorLafferty2002} 
{\sc Kondor, R.I., Lafferty, J.}
\newblock {Diffusion kernels on graphs and other discrete input spaces.}
\newblock in {\em Proc. of the 19th. Intern. Conf. on Machine Learning ICML02\/} (2002), 315-322.

\bibitem{Li2003}
{\sc Li, B.-R.}
\newblock {\em Real operator algebras}. Fields Institute Monographs, 
\newblock {World Scientific Publishing}, Singapore, 2003.

\bibitem{Linderman2018}
{\sc Linderman, G.C. and Steinerberger, S.}
\newblock {Numerical Integration on Graphs: where to sample and how to weigh.}
\newblock {arXiv:1803.06989} (2018).

\bibitem{Mhaskar1999}
{\sc Mhaskar, H., Narcowich, F.J., Prestin, J. and Ward, J.D.}
\newblock {$L^p$-Bernstein estimates and approximation by spherical basis functions.}
\newblock {Mathmatics of Computation 79}, 271 (2010), 1647--1679

\bibitem{Narang2013} 
{\sc Narang, S.K., Gadde, A., and Ortega, A.}
\newblock {Signal processing techniques for interpolation in graph structured data.}
\newblock in {\em Speech and Signal Processing, 2013 IEEE International Conference on Acoustics\/}, Vancouver, BC (2013), 5445--5449.

\bibitem{Ortega2018} 
{\sc Ortega, A., Frossard, P., Kovačević, J., Moura, J.M.F. and Vandergheynst, P.}
\newblock {Graph Signal Processing: Overview, Challenges, and Applications.}
\newblock in {\em Proceedings of the IEEE 106}, 5, (2018), 808--828.

\bibitem{Pesenson2008}
{\sc Pesenson, I.Z.}
\newblock {Sampling in Paley-Wiener spaces on combinatorial graphs.} 
\newblock {\em Trans. Amer. Math. Soc. 360}, 10 (2008), 5603–-5627.

\bibitem{Pesenson2009}
{\sc Pesenson, I.Z.}
\newblock {Variational Splines and Paley-Wiener Spaces on Combinatorial Graphs.} 
\newblock {\em Constr. Approx. 29}, 1 (2009), 1--21.

\bibitem{Pesenson2011}
{\sc Pesenson, I.Z., Pesenson, M.Z. and F\"uhr, H.}
\newblock {Cubature formulas on combinatorial graphs.} 
\newblock {\em arXiv:1104.0963 (math.FA)} (2011).

\bibitem{PuschelMoura2008}
{\sc Puschel, M., and Moura, J.M.F.}
\newblock {Algebraic Signal Processing Theory: Foundation and 1-D Time.} 
\newblock {\em IEEE Transactions on Signal Processing 56}, 8 (2008), 3572--3585.

\bibitem{Romero2017}
{\sc Romero, D., Ma, M., and Giannakis, G.B.}
\newblock {Kernel-Based Reconstruction of Graph Signals.} 
\newblock {\em IEEE Transactions on Signal Processing 65}, 3 (2017), 764--778.

\bibitem{Sandryhaila2013}
{\sc Sandryhaila, A., and Moura, J.M.F.}
\newblock {Discrete Signal Processing on Graphs} 
\newblock {\em IEEE Transactions on Signal Processing 61}, 7 (2015), 1644--1656.

\bibitem{Sasvari1994}
{\sc Sasvári, Z.}
\newblock {\em Positive Definite and Definitizable Functions.}
Akademie Verlag, Berlin, 1994.

\bibitem{Scha98} 
{\sc Schaback, R.}
\newblock {Native Hilbert spaces for radial basis functions I}. 
\newblock In {\em M. W. M\"uller et. al., eds., New Developments in Approximation Theory. 2nd International Dortmund Meeting (IDoMAT '98), vol. 132}, Int. Ser. Numer. Math., Birh\"auser Verlag, Basel (1999), 255--282.

\bibitem{SchabackWendland2003}
{\sc Schaback, R. and Wendland, H.} 
\newblock {Approximation by Positive Definite Kernels.}
\newblock In {\em Advanced Problems in Constructive Approximation}, 
Birkh\"auser Verlag, Basel (2003), 203--222

\bibitem{Schoelkopf2002} 
{\sc Sch\"olkopf, B. and Smola, A.}
\newblock {\em Learning with Kernels}.
MIT Press, Cambridge, 2002.

\bibitem{Schoenberg1942}
{\sc Schoenberg, I.J.}
\newblock {Positive definite functions on spheres} 
\newblock {\em Duke Math. J. 9}, (1942), 96--108.

\bibitem{shuman2012}
{\sc Shuman, D.I., Ricaud, B., and Vandergheynst, P.}
\newblock {A windowed graph {F}ourier transform.}
\newblock in {\em Proc. 2012 IEEE Statistical Signal Processing Workshop
  (SSP)}, (2012), 133--136.

\bibitem{shuman2016}
{\sc Shuman, D.I., Ricaud, B., and Vandergheynst, P.}
\newblock {Vertex-frequency analysis on graphs.} 
\newblock {\em Appl. Comput. Harm. Anal. 40}, 2 (2016), 260--291.

\bibitem{SmolaKondor2003}
{\sc Smola, A. and Kondor. R.}
\newblock {Kernels and Regularization on Graphs.} 
\newblock In {\em Learning Theory and Kernel Machines}, Springer Berlin Heidelberg (2003), 144--158.

\bibitem{StankovicDakovicSejdic2019}
{\sc Stankovi{\'c}, L., Dakovi{\'c}, L., and Sejdi{\'c}, E.}
\newblock {Introduction to Graph Signal Processing.} 
\newblock In {\em Vertex-Frequency Analysis of Graph Signals}, Springer, (2019), 3--108.

\bibitem{Stewart1976}
{\sc Stewart, J.}
\newblock {Positive definite functions and generalizations, an historical survey.} 
\newblock {\em The Rocky Mountain Journal of Mathematics 6}, 3 (1976), 409–-434.

\bibitem{TBL2016}
{\sc {Tsitsvero}, M., {Barbarossa}, S., and {Di Lorenzo}, P. }
\newblock {Signals on Graphs: Uncertainty Principle and Sampling.} 
\newblock{\em IEEE Trans. Sign. Proc. 64}, 18 (2016), 4845--4860.

\bibitem{Ward2018interpolating}
{\sc Ward, J.P., Narcowich, F.J., and Ward, J.D.},
\newblock{Interpolating splines on graphs for data science applications.}
\newblock{arXiv:1806.10695 (math.NA)} (2018).

\bibitem{We05} 
{\sc Wendland, H.}
\newblock {\em Scattered Data Approximation},
\newblock Cambridge University Press, Cambridge, 2005.

\end{thebibliography}
\end{document}